\documentclass[pdflatex,sn-basic,iicol]{sn-jnl}

\usepackage{amsmath}
\usepackage{amssymb}
\usepackage{amsthm}

\usepackage{newtx}
\usepackage{microtype}
\usepackage{cleveref}
\usepackage{thmtools}
\usepackage{booktabs}
\usepackage{algorithm}
\usepackage{algpseudocode}
\usepackage[bold,small]{complexity}

\theoremstyle{thmstyleone}
\newtheorem{theorem}{Theorem}
\newtheorem{lemma}[theorem]{Lemma}
\newtheorem{corollary}[theorem]{Corollary}
\newtheorem{proposition}[theorem]{Proposition}
\theoremstyle{thmstylethree}
\newtheorem{definition}[theorem]{Definition}
\newtheorem{algo}[theorem]{Algorithm}

\newcommand{\FDDS}{\mathbb{D}}

\newcommand{\lct}{<_{\mathrm{ct}}}
\newcommand{\lects}{<_{\mathrm{ct}}}
\newcommand{\lect}{\le_{\mathrm{ct}}}
\newcommand{\lp}[1][]{\prec_{#1}}
\newcommand{\lep}[1][]{\preceq_{#1}}
\newcommand{\set}[2][]{#1\{#2#1\}}
\DeclareMathOperator{\lcm}{lcm}
\DeclareMathOperator{\alcm}{alcm}

\newcommand{\minlen}{\ell}
\newcommand{\minLength}[1]{\minlen(#1)}
\newcommand{\cycle}[1]{C_{#1}}
\newcommand{\subms}{\le}
\newcommand{\aLcm}[2]{\alcm_{#1}#2}

\newcommand{\forest}[1]{\mathbf{#1}}

\newcommand{\talltrees}[2]{\Gamma_{#2}(#1)}
\newcommand{\cut}[2]{\mathscr{C}_{#2}(#1)}
\newcommand{\cutUn}[1]{\mathscr{C}_{n}(\unroll{#1})}
\newcommand{\compslendiv}[1][]{\Lambda_{#1}}
\newcommand{\unroll}[1]{\mathcal{U}(#1)}
\newcommand{\depth}[2][]{d_{#1}(#2)}
\newcommand{\setLength}[1]{L(#1)}
\newcommand{\lepref}{\lep[P]}
\newcommand{\leprefs}{\lp[P]}
\newcommand{\pref}[2]{\mathcal{\pi}_{#2}(#1)}
\newcommand{\setDive}[2]{\Lambda_{#2}(#1)}
\newcommand{\prodC}[1]{\prod_{#1}}
\newcommand{\prefE}[2]{\mathcal{P}_{#2}(#1)}
\newcommand{\bigo}[1]{\mathcal{O}(#1)}
\newcommand{\qi}[1]{p_i^{#1}}

\newcommand{\qibi}{\qi{b_i}}

\newcommand{\qibiP}[1]{\qi{2^{#1} (b_i - a_i)}}
\newcommand{\qiTbi}{\qi{2 b_i}}

\newtheorem{Example}[theorem]{Example}{\bfseries}{\itshape}

\begin{document}

\title{\bfseries Injective and pseudo-injective polynomial equations:\linebreak From permutations to dynamical systems}

\author{\fnm{Antonio E.} \sur{Porreca}}\email{antonio.porreca@lis-lab.fr}
\author*{\fnm{Marius} \sur{Rolland}}\email{marius.rolland@lis-lab.fr}

\affil{\orgname{Aix-Marseille Université, CNRS, LIS}, \orgaddress{\city{Marseille}, \country{France}}}

\abstract{We study the computational complexity of decomposing finite discrete dynamical systems (FDDSs) in terms of the semiring operations of alternative and synchronous execution, which is useful for the analysis of discrete phenomena in science and engineering. More specifically, we investigate univariate polynomials of the form~$P(X) = B$, that is with a constant side, first over the subsemiring of permutations and then over general FDDSs. We find a characterization of injective polynomials~$P$ and efficient algorithms for solving the associated equations. Then, we introduce the more general notion of pseudo-injective polynomial, which is based on a condition on the lengths of the limit cycles of its coefficients, and prove that the corresponding equations are also solvable efficiently. These results also apply even when permutations are encoded in an exponentially more compact way.}


\maketitle

\section{Introduction}
\label{sec:introduction}

Finite discrete dynamical systems (FDDSs) are, from a mathematical point of view, simply (transition) functions on a finite domain of states. As the term ``dynamical'' implies, we are interested in the evolution of the states of the system under iterated applications of the transition function. The domain being finite, any such trajectory, after a transient phase, eventually ends up in a limit cycle. The study of the limit cycles (and, in particular, the fixed points, cycles of length~$1$) of a FDDS is of great interest for applications, for instance when considering biological phenomena modeled as Boolean automata network.

When the FDDSs being analyzed become large, this analysis may become computationally intensive, and decomposing the FDDSs into smaller ones is often desirable, with the goal of establishing the global behavior from that of its subsystems. Several forms of decomposition are possible, but here we follow the approach pioneered by~\cite{dorigatti2018polynomial} and consider the category-theoretic operations of sum and product in the category of FDDSs, which correspond to their alternative and synchronous parallel execution, respectively. These endow FDDSs with a semiring algebraic structure, which allows us to express several decomposition problems in terms of (possibly multivariate) polynomial equations~$P(\vec{X}) = Q(\vec{X})$.

As already proved in~\cite{dorigatti2018polynomial}, these equations are undecidable in their general form by reduction from Hilbert's tenth problem; they become decidable with an~$\NP$ upper bound when one of the two sides of the equation is constant ($P(\vec{X}) = B)$, although this is~$\NP$-complete even in the linear case for multivariate polynomials~\citep{bridoux2020dds_semiring}.

However, several interesting restricted subclasses of equations have been found to admit polynomial-time solution algorithms. \cite{divisions_par_premier} prove that $k$-th root extraction ($X^k = B$) can be carried out in polynomial time if~$B$ is a permutation (which only consists in limit cycles, without transients).
Linear, univariate equations ($AX=B$) can be solved efficiently if~$A$ and~$B$ are dendrons (connected systems with a fixed point), as showed by~\cite{article_arbre}, or if~$A$ is a cycle and~$B$ several copies of the same cycle, even if their lengths are expressed in binary~\citep{decision_basi_eq}, or again if~$B$ is a permutation where all cycles have powers of the same prime number as lengths~\citep{divisions_par_premier}. Furthermore,  $P(X) = B$~can be solved in polynomial time if the coefficients of~$P$ and~$B$ only consist in sums of fixed points, which corresponds to equations over the natural numbers~\citep{divisions_par_premier}.

In this paper we generalize several of the above-cited results. First of all, we prove that~$P(X) = B$ can be solved efficiently whenever~$P$ is an injective polynomial, together with a simple characterization of this class of polynomials: these are exactly the polynomials where a fixed point appears in one of its coefficients (except that of the constant monomial). This condition was proved to be sufficient, but not proved to be necessary, for arbitrary (not necessarily functional) digraphs by~\cite{inj_poly_Lov,Lov2}, while~\cite{cancelable,article_arbre} characterize cancelable FDDSs (i.e., injectivity for the restricted case of polynomials of the form~$AX$) in terms of the same condition.

We also extend the notion of injectivity to pseudo-injectivity: all coefficients of~$P$, excluding the constant monomial, must contain cycles of length multiple of the shortest one. We then prove that pseudo-injective polynomial equations with a constant side also admit a polynomial-time solution algorithm.

The rest of the paper is organized as follows: after introducing some preliminaries (\autoref{sec:preliminaries}), we first consider polynomials over the permutations (\autoref{sec:poly-permutations}), a sufficient condition for their injectivity and a necessary one which also holds for arbitrary FDDSs, as well as pseudo-injective polynomials over the permutations, together with efficient algorithms for the corresponding equations. In~\autoref{sec:compact} we prove that these equations can be solved efficiently even if the permutations are encoded compactly in binary.
In~\autoref{sec:unrolls} we recall the notion of unroll of FDDSs, which allows us to define efficient algorithms that take their transient behavior into account, and prove that arbitrary polynomial equations over the unrolls can also be solved efficiently.
This is exploited in~\autoref{sec:beyond-permutations} to give efficient algorithms for pseudo-injective equations over arbitrary FDDSs, for which we also give a characterization of injective polynomials.
Finally, in~\autoref{sec:conclusions} we suggest several open problems for further research.

This article is an extended version of the conference papers~\cite{Porreca2025a,Porreca2025b}, including new results (notably~\autoref{lemma:find_tree} and~\autoref{thm:fast-alcm}) and new tools (\autoref{def:lepref}) which allow us to describe a unified algorithm and proof technique schema for all problems investigated here, which is more intuitive and more modular in terms of presentation, with simplified proofs.

\section{Preliminaries}
\label{sec:preliminaries}

A \emph{finite deterministic dynamical system} (FDDS) is a pair $(A, f_A)$ where $A$ is a finite set of states and $f_A : A \to A$ is a total function called the \emph{transition function}.
The set of FDDSs \emph{up to isomorphism}, with mutually exclusive alternative execution as \emph{sum} and synchronous parallel execution as \emph{product}, forms a semiring~\citep{dorigatti2018polynomial}, denoted by~$\mathbb{D}$. 
The two operations above define a notion of algebraic decomposition of FDDSs.

The semiring of functional digraphs, with disjoint union for sum and direct product for multiplication, is isomorphic to $\mathbb{D}$.
We recall~\citep{hammack2011graph_product_book} that the direct product of two graphs $A,B$ is the graph $C$ such that $V(C) = V(A) \times V(B)$ and
{\small
\[
E(C) = \{((u,u'), (v,v')) \mid (u,v) \in E(A), (u',v') \in E(B)\}.
\]}
An example of the product of two FDDSs is depicted in~\autoref{fig:product}.
As for natural numbers, we sometimes employ the partial operation of subtraction, by writing~\mbox{$A- B = C$} if and only if~$A = B + C$.

Then, an FDDS can be seen as the sum of connected components, where each component consists of a cycle, representing the periodic behavior, and in-trees (trees directed from leaf to root) rooted in the cycle, representing the transient behavior.
In this paper, trees are represented by minuscule bold letters (for example~$\forest{t}$) and forests (sums of trees) by majuscule bold letters (for example~$\forest{F}$). 

We can also identify two sets of special FDDSs.
First, the sets of FDDSs without transient nodes, which correspond to \emph{permutations}; in other words, a permutation $A$ is a sum of cycles.
In symbols, if~$\cycle{a}$ is the cycle of length $a$ then $A = \cycle{a_1} + \cdots + \cycle{a_m}$ with $a_1,\ldots,a_m$ the cycle lengths in $A$. 
The second kind of special FDDSs are those where each component is a \emph{dendron}, that is to say a connected component with a cycle of length~$1$.
Note that both types of FDDSs are subsemirings of~$\mathbb{D}$.
 
Given two FDDSs $A$ and~$B$, if there exists another FDDS $C$ such that~$A + C = B$ then we say that~$A$ is a submultiset of~$B$ (more precisely, the connected components of~$A$ are a submultiset of the connected components of~$B$).

\begin{figure*}[t]
	\centering
	\includegraphics[page=2,scale=1.2]{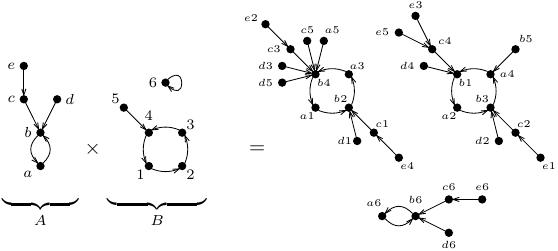}
	\caption{Product of two FDDSs~$A$ and~$B$, where the states are given a temporary label in order to show how the result is computed (for brevity, the vertices~$(u,v)$ of the product are denoted here by~$uv$). Remark that~$B$ is cancelable (state $6$~is a fixed point) and~$AB$ is pseudo-cancelable ($\setLength{AB} = \{2,4\}$ and~$\minLength{AB} = \min (\setLength{AB}) = 2$). $A$~is also trivially pseudo-cancelable, since it is connected ($\setLength{A} = \{2\}$ and~$\minLength{A}=2$).}
	\label{fig:product}
\end{figure*}

The structure of the product is algebraically rich.
For example, this semiring does not admit unique factorizations into irreducibles, that is to say, there exist four irreducible FDDSs $A,B,C,D$ such that $A \neq C$ and $A \neq D$ but $AB = CD$; furthermore, polynomial equations of the form~$P(X)=B$ may have multiple solutions (in other words, not all nontrivial polynomials are injective).
Moreover, while the periodic behavior of a product is easy to analyze, its transient behavior demands much more work.
Indeed, the product of two connected FDDS $A,B$, having cycles of length~$p$ and~$q$ respectively,
generates $\gcd(p,q)$ connected components, each with a cycle of size $\lcm(p,q)$~\citep{hammack2011graph_product_book}.

In this paper, we will give several efficient algorithms for solving the equation $P(X) = B$, with $P$ a polynomial over FDDSs and $B$ an FDDS.
More precisely, the equation will be solved in polynomial time with respect to the size (number of states or vertices)~$|B|$ of~$B$, since this bounds the size~$|P(X)|$ of the left-hand side, allowing us to avoid potential super-polynomial computations (for instance, if the degree of~$P$ is exponential with respect to~$|B|$).


\section{Polynomials over the permutations}
\label{sec:poly-permutations}

In this section we show that a polynomial over the permutations is injective if and only if at least one of its non-constant coefficients contains a fixed point.
We begin by showing that this is a sufficient condition.
The proof will be constructive and will allow us to formulate an efficient algorithm for solving~$P(X) = B$ when~$P$ respects the above constraint and~$B$ is a permutation; we first focus on polynomials of the form~$X^k$ (that is, on computing~$k$-th roots) and then generalize.
We will also show that this condition for injectivity is necessary for arbitrary polynomials over the FDDSs, and not only over the permutations.

\subsection{Roots over the  permutations}
\label{sec:roots}

In this section we consider polynomials~$P$ over the permutations of the form~$P(X) = X^k$ for some positive integer~$k$, i.e., $k$-th roots.
Note that~\cite{riva2022thesis,divisions_par_premier} already give a different efficient algorithm for this class of equations;
however, the advantage of our algorithm is that it can be easily extended to more general injective polynomials~$P$.  

In the following proofs, we process the cycles of~$X$ according to their lengths. 
More precisely we consider the \emph{smallest cycle length} in~$X$, denoted by~$\minLength{X}$, and the submultiset of connected component of~$X$ having a cycle \emph{of length dividing~$p > 0$}, denoted~$\compslendiv[p](X)$.
The first notion is important because we construct the solution starting from~$\minLength{X}$, while the second notion identifies which connected component of~$A$ and~$B$ are susceptible to generating a connected component with cycle length~$p$ in~$AB$. 
This assertion is detailed in the next lemma.

\begin{lemma}
$\compslendiv[p] \colon \FDDS \to \FDDS$ is a semiring endomorphism.
\end{lemma}

\begin{proof}
Since~$1$ has a cycle of length~$1$, we have~$\compslendiv[p](1) = \mathbf{1}$; furthermore~$\compslendiv[p](0) = \mathbf{0}$.
Since the sum in~$\FDDS$ is the disjoint union, we have~$\compslendiv[p](A+B) = \compslendiv[p](A) + \compslendiv[p](B)$.

In order to prove that~$\compslendiv[p](AB) = \compslendiv[p](A) \compslendiv[p](B)$ it suffices to consider the case of connected~$A$ and~$B$, since the product distributes over sums.
All connected components of~$AB$ have cycles of length~$\lcm \set{\minlen(A), \minlen(B)}$.
Hence, either~$\lcm \set{\minlen(A), \minlen(B)}$ divides~$p$ and~$\compslendiv[p](AB) = AB$, or~$\compslendiv[p](AB) = \mathbf{0}$.
But~$\lcm \set{\minlen(A), \minlen(B)}$ divides~$p$ if and only if~$\minlen(A)$ and~$\minlen(B)$ divide~$p$, and thus in both cases we have~$\compslendiv[p](AB) = \compslendiv[p](A) \compslendiv[p](B)$.
\end{proof}

The main idea of our algorithm for solving~$X^k = B$ is to consider the cycles of~$B$ in the following order. 

\begin{definition}[$\lect$ for permutations]
	Let~$A$ and~$B$ be connected FDDSs.
	We say that~$A \lect B$ if and only if~$\minlen(A) \le \minlen(B)$.
\end{definition}

\begin{algo}
  	\label{alg:root-permutation}
  	Let~$B$ be a permutation and~$k>0$.
  	Set~$X = \mathbf{0}$.
  	While~$X^k \subms B$, repeatedly take the shortest cycle in~$B-X^k$ and add it to~$X$.
  	If at any moment~$|X^k| > |B|$, or if~$X^k \not\subms B$, then~$B$ does not admit a~$k$-th root.
  	Otherwise, return the final value of~$X$.
\end{algo}

\begin{figure}[t]
\begin{center}
\begin{tabular}{cc}
  \toprule
~$B - X^k$                                 &~$X$             \\
  \midrule                      
~$2\cycle{2} + 12\cycle{3} + 26 \cycle{6}$ &~$\cycle{2}$                           \\
~$12 \cycle{3} + 26 \cycle{6}$             &~$\cycle{2} + \cycle{3}$               \\
~$9 \cycle{3} + 24 \cycle{6}$              &~$\cycle{2} + 2\cycle{3}$              \\
~$22 \cycle{6}$                            &~$\cycle{2} + 2 \cycle{3} + \cycle{6}$ \\
  \bottomrule
\end{tabular}
\end{center}
\caption{Example of a run of~\autoref{alg:root-permutation} on inputs~$B = 2\cycle{2} + 12\cycle{3} + 26 \cycle{6}$ and~$k =2$.}\label{fig.run_root}
\end{figure}

\autoref{fig.run_root} gives an example of execution of this algorithm.
Let us show that~\autoref{alg:root-permutation} is correct and efficient with respect to the size of~$B$. 
The main argument for its correctness is a reasoning induction over the number of element in a certain submultiset of~$X$:

\begin{definition}[Prefix and super-prefix]
\label{def:prefix}
  Let~$X = X_1 + \cdots + X_m$ be an FDDS with connected components sorted
  according to~$\lect$.
  The \emph{prefix} of length~$i$ of~$X$, denoted~$\pref{X}{i}$, is defined
  as~$\pref{X}{0} = \mathbf{0}$, as~$\pref{X}{i} = X$ whenever~$i \ge m$, and
  as~$X_1 + \cdots + X_i$ otherwise.

  Now suppose~$X = \sum_{j=1}^{r} x_j X_j$, where each~$X_j$ is a distinct
  connected component with multiplicity~$x_j$.
  The \emph{super-prefix} of length~$i$ of~$X$, denoted~$\prefE{X}{i}$, is
  defined as~$\prefE{X}{0} = \mathbf{0}$, as~$\prefE{X}{i} = X$ whenever~$i \ge r$,
  and as~$x_1 X_1 + \cdots + x_i X_i$ otherwise.
\end{definition}

\begin{lemma}
\label{thm:minlen}
Let~$X = X_1 + \cdots + X_m$ be an FDDS with connected components~$X_1 \lect \cdots \lect X_m$, and let~$k > 0$ and~$i \ge 0$.
Then~$\minlen(X^k - (\pref{X}{i})^k) = \minlen(X_{i+1})$.
\end{lemma}

\begin{proof}
Remark that~$X - (\pref{X}{i})$ divides~$X^k - (\pref{X}{i})^k$ for all~$k > 0$.
Since the product by a nonzero FDDS does not decrease the length of the cycles, the minimality of~$\minlen(X_{i+1})$ in~$X - (\pref{X}{i})$ implies~$\minlen(X^k - (\pref{X}{i})^k) \ge \minlen(X_{i+1})$.

In order to prove the reverse inequality, it suffices to show that~$X^k - (\pref{X}{i})$ contains a connected component with a cycle of length~$\minlen(X_{i+1})$.
This is the case for all components of~$X_{i+1}^k$, which are terms of the multinomial expansion of~$X^k - (\pref{X}{i})^k$.
\end{proof}

\begin{theorem}
\label{thm:root-permutation}
\autoref{alg:root-permutation} computes~$k$-th roots of permutations in polynomial time.
\end{theorem}

\begin{proof}
The algorithm only returns an FDDS~$X$ if it satisfies~$X^k = B$.
Hence, in order to show its correctness we only need to prove that if~$Y = \sqrt[k]{B}$, then the algorithm does indeed return~$Y$.
Suppose that~$Y = Y_1 + \cdots + Y_m$ with connected components~$Y_1 \lect \cdots \lect Y_m$.
We prove by induction that, at the end of the~$i$-th iteration of the algorithm, we have~$X = Y_1 + \cdots + Y_i$.
This is trivially the case when~$i=0$, hence suppose that this is the case after~$i$ iterations.
Since~$X^k \le B$, the difference~$B - X^k$ is well defined and, by~\autoref{thm:minlen}, the shortest cycle in~$B - X^k$ is~$Y_{i+1}$.
Hence, at the end of the next iteration we have~$X = Y_1 + \cdots + Y_{i+1}$ as expected.

We can now analyze the runtime of the algorithm.
The product of two FDDSs can be computed in polynomial time with respect to the size of the output, which is always bounded by the input size~$|B|$ here.
As a consequence,~$k$-th powers can also be computed in polynomial time (even with the naive algorithm, as~$k \le \log_2 |B|$).
Since comparisons and well-defined subtractions of multisets can also be computed in polynomial time, the results follows.
\end{proof}

\subsection{Sufficient condition for the injectivity of polynomials over the permutations}
\label{sec:permutations}

In order to extend~\autoref{alg:root-permutation} from equations~$X^k = B$ to arbitrary polynomials over the permutations where a non-constant coefficient contains a fixed point, we begin by generalizing~\autoref{thm:minlen}. 

\begin{lemma}
\label{lemma:poly-prefix}
Let~$P = \sum_{i=0}^k A_iX^i$ be an injective polynomial over the FDDSs.
Let~$X = X_1 + \cdots + X_m$ be an FDDS with connected components~$X_1 \lect \cdots \lect X_m$.
Finally, let~$0 \le i < m$.
Then,~$\minlen(X_{i+1}) = \minlen(P(X) - P(\pref{X}{i})) = \minlen(X_{i+1})$.
\end{lemma}

\begin{proof}
Remark that~$\minlen(A_j(X^j-(\pref{X}{i})^j))$ is greater than or equal to~$\minlen(A_j)$ and to~$\minlen(X^j-(\pref{X}{i})^j)$ for all~$j$.
Hence, by~\autoref{thm:minlen}, we have~$\minlen(A_j(X^j-(\pref{X}{i})^j)) \ge \minlen(X_{i+1})$.
Since
\begin{align*}
  P(X) - P(\pref{X}{i}) = \sum_{j=1}^m A_j(X^j-(\pref{X}{i})^j)
\end{align*}
we deduce that~$\minlen(P(X) - P(\pref{X}{i})) \ge \minlen(X_{i+1})$.

For the reverse inequality, since some~$A_j$ with~$j>0$ contains a fixed point, the term~$A_j(X^j-(\pref{X}{i})^j)$ contains a cycle of length~$\minlen(X_{i+1})$ and thus~$\minlen(P(X) - P(\pref{X}{i})) \le \minlen(X_{i+1})$.
\end{proof}

This lemma allows us to extend~\autoref{alg:root-permutation} in the following way:

\begin{algo}
\label{alg:poly-permutation}
Let~$P$ be a polynomial over the permutations where at least one non-constant coefficient contains a fixed point.
Set~$X = \mathbf{0}$.
While~$P(X) \subms B$, repeatedly take the shortest cycle in~$B-P(X)$ and add it to~$X$.
If at any moment~$|P(X)| > |B|$, or if~$P(X) \not\subms B$, then~$P(X)=B$ does not admit a solution.
Otherwise, return the final value of~$X$.
\end{algo}

The correctness and efficiency of~\autoref{alg:poly-permutation} is guaranteed by an argument analogous to the proof of of~\autoref{thm:root-permutation}, with~\autoref{lemma:poly-prefix} playing the role of~\autoref{thm:minlen}.

\begin{theorem}
\autoref{alg:poly-permutation} solves injective polynomial equations over the permutations in polynomial time. \qed
\end{theorem}

In addition, \autoref{lemma:poly-prefix} also allows us to prove a sufficient condition for the injectivity of polynomials over the permutations.
 
\begin{proposition}\label{prop:inj_permutation}
Let~$P = \sum_{i=0}^k A_iX^i$ be a polynomial over the permutations such that the coefficient of at least one non-constant monomial contains a fixes point.
Then~$P$ is injective over the semiring of permutations.
\end{proposition}

\begin{proof}
Let~$X = X_1 + \cdots + X_m$ and~$Y = Y_1 + \cdots + Y_n$ be two permutations with connected components~$X_1 \lect \cdots \lect X_m$ and~$Y_1 \lect \cdots \lect Y_n$ and such that~$P(X) = P(Y)$.
We prove that~$X = Y$ by induction on the prefixes, that is,~$\pref{X}{i} = \pref{Y}{i}$ for all~$i$.
This is trivially the case for~$i=0$.

In order to prove that~$\pref{X}{i+1} = \pref{Y}{i+1}$, by induction hypothesis it suffices to show~$X_{i+1} = Y_{i+1}$, that is~$\minlen(X_{i+1}) = \minlen(Y_{i+1})$.
Since~$P(X) - P(\pref{X}{i}) = P(Y) - P(\pref{Y}{i})$, \autoref{lemma:poly-prefix} indeed implies~$\minlen(X_{i+1}) = \minlen(P(X) - P(\pref{X}{i})) = \minlen(P(Y) - P(\pref{Y}{i}) = \minlen(Y_{i+1})$ and thus~$X_{i+1} = Y_{i+1}$.
\end{proof}

\subsection{Necessary condition for the injectivity of polynomials over FDDSs}

In this section we show that any polynomial over FDDSs where no non-constant coefficient contains a dendron is not injective.
For this proof we construct two distinct permutation~$X$ and~$Y$ such that~$P(X) = P(Y)$, which will also prove that any polynomial \emph{over the permutations} where no non-constant coefficient contains a fixed point is not injective.
This mean that~\autoref{alg:poly-permutation} actually allows us to solve the equation~$P(X) = B$ if~$P$ is an injective polynomial over the permutations.
In order to construct the permutations~$X$ and~$Y$, we summarize the construction introduced in the proof of Lemma~33 of~\cite{article_arbre} and also in~\cite{cancelable}.
Our first goal is to prove that if~$A$ does not contain a dendron then~$AX^k = AY^k$ for all integer~$k>0$.
 
Let~$\delta_J$ be a sequence inductively defined by
\begin{align*}
  \begin{cases}
    \delta_\varnothing &= 1 \\
    \delta_{J \cup \{a\}} &= \gcd(a, \lcm(J)) \delta_J
  \end{cases}
\end{align*}
for all nonnegative integer~$a \notin J$; notice that if~$J = \varnothing$ then~$\lcm(J) = 1$. 
Let~$N$ be a set of integers greater than~$1$.
For each subset~$I \subseteq N$, we define
\begin{align*}
  \alpha_I &= \delta_N\prod_{a \in N} a \\
  \beta_I &= \alpha_I + (-1)^{|I|}\delta_I \prod_{a \in N - I} a.
\end{align*}
An example of these sequences is given in~\autoref{table:alpha_beta_delta}.
We exploit~$\alpha_I$ and~$\beta_I$ in order to construct the permutations~$X$ and~$Y$ mentioned above.
Let us first assume that $A$ is a cycle of length~$>1$.

\begin{figure}[t]
\begin{center}
\begin{tabular}{ccccc}
  \toprule
  $I$         &~$\alpha_I$ &~$\delta_I$ &~$\prod_{a \in N - I}$ &~$\beta_I$ \\
  \midrule
  $\varnothing$ & 216       &~$1$        & 36                   & 252      \\
  2           & 216       &~$1$        & 18                   & 198      \\
  3           & 216       &~$1$        & 12                   & 204      \\
  6           & 216       &~$1$        & 6                    & 210      \\
  2,3         & 216       &~$1$        & 6                    & 222      \\
  2,6         & 216       &~$2$        & 3                    & 222      \\
  3,6         & 216       &~$3$        & 2                    & 222      \\
  2,3,6       & 216       &~$6$        & 1                    & 210      \\
  \bottomrule
\end{tabular}
\end{center}
\caption{Table of values of the sequences~$\alpha, \beta$ and~$\delta$ for the subsets of~$\{2,3,6\}$. 
		The column with head~$I$ gives the subsets, and the columns of head~$\alpha_I, \delta_I$ and~$\beta_I$ give the value of the respective sequence for~$I$.}\label{table:alpha_beta_delta}
\end{figure}

\begin{lemma}\label{lemme:cycleNonInjTech}
	Let~$k\ge 1$ be an integer.
	Let~$N$ be a subset of nonnegative integer and let~$a>1$ be an element of~$N$.
	Let~$I$ be a subset of~$N - \{a\}$ and~$J = I \cup \{a\}$.
	Then
        \begin{align*}
          C_a (\beta_I C_{\lcm(I)} + \beta_J C_{\lcm(J)})^k = C_a (\alpha_I C_{\lcm(I)} + \alpha_J C_{\lcm(J)})^k.
        \end{align*}
\end{lemma}

\begin{proof}
From the proof of Lemma~33 of~\cite{article_arbre}, we have
\begin{align*}
  C_a (\beta_I C_{\lcm(I)} + \beta_J C_{\lcm(J)}) = C_a (\alpha_I C_{\lcm(I)} + \alpha_J C_{\lcm(J)}).
\end{align*}
Thus, we can replace an instance~$C_a (\beta_I C_{\lcm(I)} + \beta_J C_{\lcm(J)})$ in the equation of the above statement and obtain
\begin{gather*}
  C_a (\beta_I C_{\lcm(I)} + \beta_J C_{\lcm(J)})^k = \\
  (\alpha_I C_{\lcm(I)} + \alpha_J C_{\lcm(J)}) C_a (\beta_I C_{\lcm(I)} + \beta_J C_{\lcm(J)})^{k-1}.
\end{gather*}
The lemma follows by inductive application of the same substitution.
\end{proof}

We then extend this construction to arbitrary permutations without fixed points.

\begin{lemma}\label{lemme:noInjMonomeP1}
	Let~$A = A_1 + \cdots + A_{m_A}$ be a permutation.
	If~$A$ does not contain fixed points then, for all~$k>0$, there exist two distinct permutations~$X$ and~$Y$ such that~$AX^k = AY^k$.
\end{lemma}

\begin{proof}
	Let $\setLength{A}$ be set of lengths of cycles in~$A$.
	Let~$X = \sum_{I \subseteq \setLength{A}} \alpha_I C_{\lcm(I)}$ and~$Y = \sum_{I \subseteq \setLength{A}} \beta_I C_{\lcm(I)}$. 
	(For example, if~$A = \cycle{2} + \cycle{3} + \cycle{6}$ then the values in~\autoref{table:alpha_beta_delta} give~$X = 216 (\cycle{1} + \cycle{2} + \cycle{3} + \cycle{6})$ while~$Y = 252 \cycle{1} + 198 \cycle{2} + 204 \cycle{3} + 876 \cycle{6}$.)

	Since for all subsets~$I$ of~$\setLength{A}$ the integer~$\beta_I$ is equal to~$\alpha_I$ plus a nonzero term, we have~$\alpha_I \ne \beta_I$, and in particular~$\alpha_\varnothing \ne \beta_\varnothing$.
	Since these two integers describe the number of fixed point in~$X$ and~$Y$, we deduce that~$X \neq Y$.
	
	Let~$a$ be an element of~$\setLength{A}$.
	Let~$I_1, \ldots, I_n$ be the list of all subsets of~$\setLength{A} - \{a\}$.
	We define~$J_1,\ldots, J_n$ as the subsets of~$\setLength{A}$ such that~$J_i = I_i \cup \{a\}$ for all~$i$.
	Then~$Y = \sum_{i=1}^n (\beta_{I_i} C_{\lcm(I_i)} + \beta_{J_i} C_{\lcm(J_i)})$.
	By distributing powers and products over the sum we obtain that $C_a Y^k$ is equal to:
	\begin{align*}
		\underset{k_1+\cdots+k_n = k}{\sum}
		&\binom{k}{k_1,\ldots,k_n} \Biggl( \Big(C_a (\beta_{I_1} C_{\lcm(I_1)} + \beta_{J_1} C_{\lcm(J_1)}\Big)^{k_1} \\&\prod_{i=2}^{n} \Big(\beta_{I_i} C_{\lcm(I_i)} + \beta_{J_i} C_{\lcm(J_i)}\Big)^{k_i}\Biggr).
	\end{align*}
	From~\autoref{lemme:cycleNonInjTech} we deduce that~$C_a (\beta_{I_1} C_{\lcm(I_1)} + \beta_{J_1} C_{\lcm(J_1)})^{k_1} = C_a (\alpha_{I_1} C_{\lcm(I_1)} + \alpha_{J_1} C_{\lcm(J_1)})^{k_1}$.
	Thus by replacing these terms in the previous equation, we obtain that $C_a Y^k$ is equal to:
	\begin{align*}
		\underset{k_1+\cdots+k_n = k}{\sum}
		&\binom{k}{k_1,\ldots,k_n}
		\Bigg(C_a \Big(\alpha_{I_1} C_{\lcm(I_1)} + \alpha_{J_1} C_{\lcm(J_1)}\Big)^{k_1} \\
		&\prod_{i=2}^{n} \Big(\beta_{I_i} C_{\lcm(I_i)} + \beta_{J_i} C_{\lcm(J_i)}\Big)^{k_i}\Bigg).
	\end{align*}
	By repeating this substitution for all~$i$ between $2$ and~$n$, it follows that~$C_a Y^k$ is equal to:
	\[C_a
	\underset{k_1+\cdots+k_n = k}{\sum}
	\binom{k}{k_1,\ldots,k_n}  \prod_{i=1}^{n} (\alpha_{I_i} C_{\lcm(I_i)} + \alpha_{J_i} C_{\lcm(J_i)})^{k_i}.\] 
	Since~$X^k = \sum_{k_1+\cdots+k_n=k}\binom{k}{k_1,\ldots,k_n}  \prod_{i=1}^{n} (\alpha_{I_i} C_{\lcm(I_i)} + \alpha_{J_i} C_{\lcm(J_i)})^{k_i}$, we conclude that~$C_a Y^k = C_a X^k$.
	
	This property is true for each connected component~$\cycle{a}$ of~$A$, therefore by summing all the terms and by making the necessary factorizations, we conclude that~$AX^k = AY^k$.
\end{proof}

Notice that the product of a generic FDDS with a permutation has a particular form. 
Indeed, if we consider an FDDS~$A = A_1 + \cdots + A_{m_A}$ and a permutation~$X = X_1 + \cdots + X_{m_X}$ such that each~$A_i$ and each~$X_j$ are connected then, for each~$i$ and~$j$, the FDDS~$A_iX_j$ has~$\gcd(\minLength{A_i},\minLength{X_j})$ connected component  with cycle length~$\lcm(\minLength{A_i},\minLength{X_j})$ and the sequence of the trees rooted in their cycle are the trees rooted in~$A_i$ repeated periodically.

\begin{lemma}[\cite{article_arbre,pas_premier,richard2025primefunctionaldigraphseiferts}]\label{lemma:prod_avec_perm}
	Let~$A$ be a connected FDDS and~$X$ a permutation.
	Then, the rooted trees in the cycle of each connected component of~$AX$ are the sequence of rooted trees in the cycle of~$A_i$ periodically repeated. \qed
\end{lemma}

\begin{corollary}[\cite{article_arbre,pas_premier,richard2025primefunctionaldigraphseiferts}]\label{cor:prod_avec_perm}
	Let~$A$ be a connected FDDS and~$X,Y$ be two permutations. 
	If~$\cycle{\minLength{A}} X = \cycle{\minLength{A}} Y$ then~$A X = A Y$.
\end{corollary}

\begin{proof}
	Suppose that~$\cycle{\minLength{A}} X = \cycle{\minLength{A}} Y$.
	Then, there exists a bijective function~$f$ between the connected components of~$\cycle{\minLength{A}} X$ and those of~$\cycle{\minLength{A}} Y$ such that~$\minLength{f(D)} = \minLength{D}$ for all connected components~$D$ of~$\cycle{\minLength{A}} X$.
	
	Now consider~$AX$ and~$AY$.
	There exist two bijective functions~$g$ and~$h$ from the set of connected components of~$AX$ and~$AY$, respectively, to the connected components of~$\cycle{\minLength{A}} X$ and~$\cycle{\minLength{A}} Y$, respectively, such that~$g(D) = \minLength{D}$ and~$h(D) = \minLength{D}$.
	
	By~\autoref{lemma:prod_avec_perm}, the rooted trees of each connected component of~$AX$ and~$AY$ are the sequence of the rooted trees in the cycle of~$A$ periodically repeated.
	Therefore, for all connected components~$D_X$ of~$AX$, the components~$D_X$ and~$D_Y = h^{-1}(f(g(D_X)))$ have the same sequence of trees rooted in their cycle.
	Furthermore, by definition,~$D_X$ and~$D_Y$ have the same cycle length.
	Therefore,~$D_X$ =~$D_Y$. 
\end{proof}

Thanks to this property we can generalize~\autoref{lemme:noInjMonomeP1} to the case where~$A$ is an FDDS without dendrons. 
This gives us a necessary condition for the injectivity of monomials over FDDSs.

\begin{proposition}\label{prop:charaInjMonome}
	Let $AX^k$ with $k>0$.
	If $A$ does not contain a dendron then $AX^k$ is not injective.
\end{proposition}

\begin{proof}
	Let~$A = A_1 + \cdots + A_{m_A}$, where each~$A_i$ is a connected component.
	Suppose that $A$ does not contain a dendron, that is, it does not contain a connected component with cycle length~$1$.
	Let $A' = A_1' + \cdots + A_{m_A}'$ be the permutation such that $\minLength{A_i'} = \minLength{A_i}$ for all $i$.
	By the proof of~\autoref{lemme:noInjMonomeP1}, there exist two different permutations $X$ and $Y$ such that $A_i' X^k = A_i'Y^k$.
	Therefore, by~\autoref{cor:prod_avec_perm}, we have that $A_i X^k = A_i Y^k$.
	By summation, $AX^k = AY^k$.
\end{proof}

\begin{corollary}
	Let~$A$ be an FDDS not containing a dendron.
	Then~$AX^k$ is not injective for any integer $k>0$. \qed
\end{corollary}

To complete the necessary condition over the injectivity, we prove that a polynomial $P$ is not injective if a monomial derived from~$P$ is not injective.
The coefficient of this monomial is just the sum of all the non-constant coefficients of the polynomial, and its degree is strictly positive.
That is, we will show that $P = \sum_{i=0}^{m}A_iX^i$ is not injective if $(A_1 + \cdots + A_m) X^k$ is not injective for any integer $k>0$.

\begin{theorem}\label{th:cond_nes_inj}
	Let~$P = \sum_{i=0}^{m} A_i X^{i}$ be a polynomial over FDDSs.
	If no non-constant coefficient of $P$ contains a dendron then $P$ is not injective.
\end{theorem}

\begin{proof}
	Consider the monomial~$M = (\sum_{i=1}^{m} A_i)X$.
	Let~$X$ and~$Y$ be the permutations constructed in the proof of~\autoref{prop:charaInjMonome}.
	As previously~$X \neq Y$ but~$A_i X^i = A_i Y^i$ for all~$i$ between~$1$ and~$m$.
	Therefore, by summation we have~$P(X) = P(Y)$.
\end{proof}

\subsection{Pseudo-injective polynomials over the permutations}
\label{sec:pseudo-injective-permutations}

By combining~\autoref{th:cond_nes_inj} and~\autoref{prop:inj_permutation}, we prove that a polynomial over the permutations is injective if and only if one of its non-constant monomials has a coefficient containing a fixed point.
From this characterization, we observe that the smallest cycle length, let us call it $g$, in the non-constant coefficients is a divisor of the set of cycle lengths of the non-constant coefficients because $g = 1$.
We call polynomials satisfying this condition, without requiring $g = 1$, \emph{pseudo-injective} polynomials.
The number $g$ will be called the \emph{seed} of the polynomial.
Remark that the pseudo-injective polynomials of seed~$1$ are exactly the injective ones.
If the polynomial is a linear monomial, that is, if~$P = A X$, we also say that the FDDS~$A$ is \emph{pseudo-cancelable}, as this property is a generalization of cancelability.

It is important to notice that these notions are only based on the form of the polynomials and not on its ``degree of injectivity''.
That is to say, they do not \emph{a priori} limit the number of different FDDSs $X_1,\ldots,X_r$ such that $P(X_1) = \cdots = P(X_r)$.

In the following we show that we can efficiently solve the equation $P(X) = B$ if $P$ is a pseudo-injective polynomial over the permutations and $B$ is a permutation. 
We begin with the simplest case: the resolution of the equation $\cycle{a} \cycle{x} = b \cycle{b}$. 
A simple method for solving these equations is to enumerate all the integers $j$ between $1$ and $b$ and check for each of them if $\cycle{a} \cycle{j} = b \cycle{b}$.
This is polynomial-time over $b$.

Let us analyze what can be the possible values of $j$.
To do this, recall that if $\cycle{a} \cycle{x} = b \cycle{b}$, then $\lcm(a,x) = b$.
Therefore, by definition of $\lcm$, if we consider the prime factorizations $\prodC{} p_i^{a_i}$ and $\prodC{} p_i^{x_i}$ of $a$ and $x$ respectively, where $p_i$ is the $i$-th prime, then $\lcm(a, x) = s_x s_a e$ with
\begin{align*}
  &s_x = \prodC{a_i < x_i} p_i^{x_i}&
  &s_a = \prodC{x_i < a_i} p_i^{a_i}&
  &e= \prodC{a_i = x_i} p_i^{a_i}.
\end{align*}
We deduce that $s_x$ is the smallest integer such that $b = \lcm(a, s_x)$.
In the following, if $a$ divides $b$, this number $s_x$ will be called the \emph{anti-$\lcm$} of $b$ with respect to $a$.

\begin{definition}
\label{def:alcm}
Let~$a, b > 0$ such that~$a$ divides~$b$.
Then, the \emph{anti-$\lcm$ of~$b$ with respect to~$a$}, in symbols~$\alcm_a b$, is the smallest~$c$ such that~$\lcm \set{a, c} = b$.

For brevity, we also define the anti-$\lcm$ of an FDDS~$B$ with respect to an FDDS~$A$ as~$\alcm_A B = \alcm_{\minlen(A)} \minlen(B)$, and its anti-$\lcm$ with respect to~$a$ as~$\alcm_a B = \alcm_a \minlen(B)$.
\end{definition}

Notice that~$\alcm_a b$ is always well-defined thanks to the hypothesis that~$a$ divides~$b$, which implies the existence of at least one~$c$ such that~$\lcm \set{a, c} = b$, namely~$b$ itself; furthermore, since all such~$c$ are natural numbers, there exist a unique minimum one.

In the following, for better readability, we shorten~$a = \prod_{i=1}^\infty p_i^{a_i}$ as~$a = \prod p_i^{a_i}$.
Furthermore, if we also have~$b = \prod p_i^{b_i}$, then we write~$\prod_{b_i>a_i} p_i^{b_i}$ for the product of the primes having larger exponent in~$b$ than in~$a$ (taken with that exponent), and similarly for~$\prod_{b_i \le a_i} p_i^{b_i}$ and~$\prod_{b_i = a_i} p_i^{b_i}$.
Finally, recall that~$\lcm \set{a, b} = \prod p_i^{\max \set{a_i, b_i}}$ and~$\gcd \set{a, b} = \prod p_i^{\min \set{a_i, b_i}}$.

\begin{theorem}
\label{thm:alcm-as-product}
Suppose that~$a$ divides~$b$ and let~$a = \prod p_i^{a_i}$ and~$b = \prod p_i^{b_i}$ be their prime factorizations.
Then~$\alcm_a b = \prod_{b_i>a_i} p_i^{b_i}$.
\end{theorem}

\begin{proof}
First of all, let us prove that we indeed have~$\lcm \set[\big]{a, \prod_{b_i>a_i} p_i^{b_i}} = b$.
We have
\begin{align*}
  \lcm \set[\Big]{a, \prod_{b_i>a_i} p_i^{b_i}}
  &= \lcm \set[\Big]{\prod p_i^{a_i}, \prod_{b_i>a_i} p_i^{b_i}} \\
  &= \lcm \set[\Big]{\prod_{b_i \le a_i} p_i^{a_i} \times \prod_{b_i>a_i} p_i^{a_i}, \prod_{b_i>a_i} p_i^{b_i}}.
\end{align*}
Since~$a$ divides~$b$, we have~$b_i \ge a_i$ for all~$i$ and thus
\begin{align*}
  \prod_{b_i \le a_i} p_i^{a_i} =
  \prod_{b_i = a_i} p_i^{a_i}
\end{align*}
which implies
\begin{align*}
  \lcm \set{a, \alcm_a b}
  &= \lcm \set[\Big]{\prod_{b_i \ge a_i} p_i^{a_i}, \prod_{b_i>a_i} p_i^{b_i}} \\
  &= \prod_{b_i=a_i} p_i^{a_i} \times \prod_{b_i>a_i} p_i^{b_i} = \prod_{b_i \ge a_i} p_i^{b_i} = b.
\end{align*}
Now suppose that~$\lcm \set{a, c} = b$ and let~$c = \prod p_i^{c_i}$ be its prime factorization.
Then~$b_i = \max \set{a_i, c_i}$ for all~$i$; hence, whenever~$b_i > a_i$ we have~$b_i = c_i$ and thus~$b_i \le c_i$;
this means that~$\prod_{b_i>a_i} p_i^{b_i}$ divides~$c$, implying~$\prod_{b_i>a_i} p_i^{b_i} \le c$ as required.
\end{proof}

This proof actually shows that~$\alcm_a b$ is not only the smallest integer having the properties required by~\autoref{def:alcm}, but also the minimum one with respect to divisibility, which is a more useful property.

\begin{theorem}
\label{thm:alcm-minimum-by-div}
The integer~$\alcm_a b$ is minimal with respect to the divisibility partial order among the~$c$ satisfying~$\lcm \set{a, c} = b$.
\qed
\end{theorem}

We can actually prove an even stronger property of the solutions of~$\lcm \set{a, c} = b$.

\begin{theorem}
\label{thm:divisor-coprime-alcm}
If~$\lcm \set{a, c} = b$, then the integer~$\frac{c}{\alcm_a b}$ is a divisor of~$a$ and coprime with~$\alcm_a b$.
\end{theorem}

\begin{proof}
Notice that~$\frac{c}{\alcm_a b}$ is indeed an integer by~\autoref{thm:alcm-minimum-by-div}.

By~\autoref{thm:alcm-as-product} and by recalling that if~$a = \prod p_i^{a_i}$,~$b = \prod p_i^{b_i}$, and~$c = \prod p_i^{c_i}$ then~$b_i = \max \set{a_i, c_i}$ for all~$i$ by a property of the~$\lcm$, we have
\begin{align*}
  \frac{c}{\alcm_a b}
  &= \frac{\prod p_i^{c_i}}{\prod_{b_i > a_i} p_i^{b_i}}
  = \frac{\prod_{b_i \le a_i} p_i^{c_i} \times \prod_{b_i > a_i} p_i^{c_i}}{\prod_{b_i > a_i} p_i^{b_i}} \\
  &= \frac{\prod_{b_i \le a_i} p_i^{a_i} \times \prod_{b_i > a_i} p_i^{b_i}}{\prod_{b_i > a_i} p_i^{b_i}}
  = \prod_{b_i \le a_i} p_i^{a_i}
\end{align*}
which is, by inspection of its factorization, a divisor of~$a$ and coprime with~$\alcm_a b$.
\end{proof}

From this, we can formulate an algorithm with polynomial runtime with respect to $\log b$ for computing $\aLcm{a}{b}$.
Indeed, the computation of $\aLcm{a}{b}$ can be obtained by reduction to a~$\gcd$ as follows.

\begin{theorem}
\label{thm:fast-alcm}
Let~$a, b > 0$ such that~$a$ divides~$b$ and let~$a = \prod p_i^{a_i}$ and~$b = \prod p_i^{b_i}$ be their prime factorizations.
Then we have~$\alcm_a b = \gcd \set[\big]{\big(\frac{b}{a}\big)^k, b}$ whenever~$k \ge b_i$ for all~$i$.
\end{theorem}

\begin{proof}
Remark that~$\frac{b}{a} = \prod_{i=0}^\infty p_i^{b_i-a_i}$ and~$(\frac{b}{a})^k = \prod_{i=0}^\infty p_i^{k(b_i-a_i)}$.
Let~$p_i$ be a prime number.
Then, either~$b_i = a_i$, implying that~$p_i$ has exponent~$0$ in~$\frac{b}{a}$ and thus in~$\gcd \set[\big]{(\frac{b}{a})^k, b}$, or~$b_i > a_i$, implying~$k(b_i-a_i) > b_i$ by the hypothesis on~$k$.
In the latter case,~$p_i$ has exponent~$b_i$ in~$\gcd \set[\big]{(\frac{b}{a})^k, b}$.
In other words, we have
\begin{align*}
  \gcd \set[\Big]{\Big(\frac{b}{a}\Big)^k, b} = \prod_{b_i>a_i} p_i^{b_i} = \alcm_a b
\end{align*}
as required.
\end{proof}

\begin{theorem}
\label{thm:complexity-alcm}
The value of~$\alcm_a b$ can be computed in polynomial time with respect to~$\log_2 b$.
\end{theorem}

\begin{proof}
Let~$k = \lceil \log_2 b \rceil$; then~$k \ge b_i$ for all~$i$ and hence is a suitable upper bound for computing~$\alcm_a b$.
The value of~$\frac{b}{a}$ can be computed~$\bigo{1}$ time and the value of~$\big(\frac{b}{a}\big)^k$ in~$\bigo{\log b}$ time under the uniform cost model.
Finally,~$\gcd \set[\big]{\big(\frac{b}{a}\big)^k, b}$ can be computed in polynomial time with respect to the number of bits of the operands, which is~$\bigo{\log^2 b}$ and~$\bigo{\log b}$ respectively.
\end{proof}

Summarizing, in order to solve $\cycle{a}\cycle{x} = b \cycle{b}$, we can just compute $\aLcm{a}{b}$ and check if $b/ \gcd(a,\aLcm{a}{b})$ is coprime with $\aLcm{a}{b}$.
Indeed, if there exists a solution then $x = q \aLcm{a}{b}$ with $q$ a divisor of $a$ coprime with $\aLcm{a}{b}$. 
Furthermore, the number of copies of $\cycle{b}$ in $\cycle{a} \cycle{x}$ is $\gcd(a,q\aLcm{a}{b}) = \gcd(a,\aLcm{a}{b}) q$.

We can intuitively think that it would be possible to extend this reasoning to the case where $A$ and/or $B$ are arbitrary permutations.
This would give us a procedure similar to the following one.

\begin{algo}
\label{alg:division-permutations}\label{algo:div_cycle}
Let~$A$ and~$B$ be permutations.
Let~$X = \mathbf{0}$.
While~$B \ne 0$:
\begin{itemize}
\item let~$a$ and~$b$ be the multiplicities of~$C_{\minlen(B)}$ in~$A C_{\alcm_A B}$ and~$B$ respectively;
\item set~$X = X + \frac{b}{a} C_{\alcm_A B}$ and~$B = B - \frac{b}{a} A C_{\alcm_A B}$.
\end{itemize}
If at any moment~$|A| > |B|$ or~$A C_{\alcm_A B}$ is not a submultiset of~$B$, then the quotient is undefined.
Otherwise, return the final value of~$X$.
\end{algo}

\begin{theorem}
\autoref{alg:division-permutations} runs in polynomial time.
\end{theorem}

\begin{proof}
All operations on permutations can be computed in polynomial time, as can~$\alcm_A B$ by~\autoref{thm:complexity-alcm}; furthermore, the number of iterations of the algorithm is bounded by the size of~$B$.
\end{proof}

However, this intuition is incorrect.

\begin{Example}\label{example:counter_example_algo}
	Consider the equation $(\cycle{2} + \cycle{3}) X = 5 \cycle{6}$. \autoref{alg:division-permutations} does not find a solution because $\aLcm{2}{6} = 3$ and $3 \cycle{3}$ is not a submultiset of $5 \cycle{6}$.
	But $(\cycle{2} + \cycle{3}) \cycle{6} = 5 \cycle{6}$.
\end{Example}

We will, however, prove that this approach does indeed work for~$AX = B$ if $A$ is pseudo-cancelable.

\begin{lemma}
\label{thm:detach-cycle}\label{lemma:first_cond_cycle}
Let~$A,X$ and $B\ne 0$ be permutations with~$A$ pseudo-cancelable.
If~$AX = B$, then~$C_{d \alcm_A B} \subms X$ for some divisor~$d$ of~$\minlen(A)$ coprime with~$\alcm_A B$.
\end{lemma}

\begin{proof}
Let~$C_b$ be a cycle of minimal length in~$B$.
Then~$C_b = C_x C_a$ for some cycles~$C_x \subms X$ and~$C_a \subms A$ and~$b = \lcm \set{x, a}$.
Since~$d = \frac{x}{\alcm_A B}$ is a divisor of~$\minlen(A)$ coprime with~$\alcm_A B$ by~\autoref{thm:divisor-coprime-alcm}, we have~$C_x = C_d C_{\alcm_A B} = C_{d\alcm_A B} \le X$.
\end{proof}

Although~\autoref{thm:detach-cycle} partially describes a cycle of~$X$, one unknown remains: what value to give to $d$.
In the following lemma, we show that we can resolve this unknown by setting the value of $d$ to $1$.
Once again, this is only possible since $A$ is pseudo-cancelable.
We have already illustrated this with~$(\cycle{2} + \cycle{3}) \cycle{6} = 5 \cycle{6}$ in~\autoref{example:counter_example_algo}.

\begin{lemma}
\label{thm:exchange-cycles}\label{lemma:second_cond_cycle}
Let~$A$ and~$X$ be permutations with~$A$ pseudo-cancelable, let~$B = AX$, and let~$d$ be a divisor of~$\minlen(A)$ coprime with~$\alcm_A B$.
Then~$A(X+C_{d \alcm_A B}) = A(X+eC_{\frac{d}{e} \alcm_A B}) = B$ for all divisors~$e$ of~$d$.
\end{lemma}

\begin{proof}
Let~$C_a$ be a cycle in~$A$.
Since~$d$ is coprime with~$\alcm_A B$, we have
\begin{align}
  \label{eq:gcds}
  \gcd \set{d \alcm_A B, a} = \gcd \set{\alcm_A B, a} \times \gcd \set{d, a}.
\end{align}
Since~$A$ is pseudo-cancelable,~$\minlen(A)$ divides~$a$; by transitivity,~$d$,~$e$, and~$\frac{d}{e}$ also divide~$a$.
Hence~$\gcd \set{d, a} = d$ and~$d = e \times \frac{d}{e} = e \times \gcd \set[\big]{\frac{d}{e}, a}$, and~\autoref{eq:gcds} becomes
\begin{align*}
  \gcd \set{d \alcm_A B, a}
  &= \gcd \set{\alcm_A B, a} \times d \\
  &= \gcd \set{\alcm_A B, a} \times e \times \gcd \set[\big]{\textstyle\frac{d}{e}, a} \\
  &= \gcd \set[\big]{\textstyle\frac{d}{e} \alcm_A B, a} \times e.
\end{align*}
Furthermore, since~$\frac{d}{e}$ divides~$a$ and is coprime with~$\alcm_A B$, we have
\begin{align*}
  \lcm \set[\big]{\textstyle\frac{d}{e} \alcm_A B, a}
  &= \lcm \set[\big]{
    \lcm \set[\big]{\textstyle\frac{d}{e}, a},
    \lcm \set{\alcm_A B, a}
  } \\
  &= \lcm \set{a, \lcm \set{\alcm_A B, a}} \\
  &= \lcm \set{\alcm_A B, a}.
\end{align*}
This implies
\begin{align*}
  C_a C_{d \alcm_A B}
  &= \gcd \set{d \alcm_A B, a} \times C_{\lcm \set{d \alcm_A B, a}} \\
  &= \gcd \set[\big]{\textstyle\frac{d}{e} \alcm_A B, a} \times e \times C_{\lcm \set{\frac{d}{e} \alcm_A B, a}} \\
  &= C_a \times e C_{\frac{d}{e} \alcm_A B}
\end{align*}
By summation on all cycles~$C_a$ of~$A$, the result follows.
\end{proof}

Remark that~\autoref{alg:division-permutations} constructs a permutation by cycle length.
However, the order of generation is not necessarily given by~$\lect$.
This is because, even for a fixed integer $g$, the function $\aLcm{g}{x}$ is not monotonically increasing. 

\begin{Example}
	For~$g = 2$, we have~$4 < 6$ but~$\aLcm{2}{4} = 4 > \aLcm{2}{6} = 3$. 
\end{Example} 

In order to solve this problem we introduce a new total order over connected components, defined both in terms of a polynomial and of~$\lect$.
This allows us to compare different related components with respect to a context.

\begin{definition}
\label{def:lepref}
Let~$P = \sum_{i=1}^k A_i X^i$ be a polynomial without constant term and let~$X$ and~$Y$ be two connected FDDSs.
Let~$B_1$ and~$B_2$ be, respectively, the smallest connected components of~$P(X)$ and~$P(Y)$ according to~$\lect$.
We say that~$X \lep[P] Y$ if and only if~$B_1 \lct B_2$, or if~$B_1 = B_2$ and~$X \lect Y$.

If~$P$ is linear, i.e., if~$P = AX$, then we also write~$\lep[A]$ for~$\lep[P]$.
\end{definition}

Remark that for two connected components~$A$ and~$B$ we can have~$A \leprefs B$ but~$B \lects A$ at the same time:

\begin{Example}
	Consider~$P = \cycle{2} X$, $A = \cycle{4}$ and~$B = \cycle{3}$.
	Then~$B \lects A$, as~$3 < 4$.
	However $A \leprefs B$ since~$\minLength{\cycle{2} \cycle{4}} = 4 < 6 = \minLength{\cycle{2} \cycle{3}}$.	
\end{Example}

By construction, the permutation constructed in~\autoref{alg:division-permutations} is sorted according to $\lep[A]$.
Therefore, we can exploit this order in the following proof.

\begin{theorem}
Let~$A$ and~$B$ be permutations with~$A$ pseudo-cancelable.
If the equation~$AX=B$ admits a solution, then~\autoref{alg:division-permutations} finds the solution maximizing the number of connected components.
\end{theorem}

\begin{proof}
Suppose that~\autoref{alg:division-permutations} returns~$X$; we prove that~$AX=B$.
Let~$X = x_1X_1 + \cdots + x_mX_m$ where all~$X_i$ are distinct connected components and~$x_i$ is their multiplicity.
Let~$B_i = B - A(\prefE{X}{i})$ for all~$i$.
Then~$x_iAX_i \subms B_{i-1}$ for all~$i \ge 1$ and~$B_m = \mathbf{0}$.
We also have
\begin{align*}
  B_i
  &= B - A(\prefE{X}{i}) \\
  &= B - A(\prefE{X}{i-1} + x_iAX_i) \\
  &= B_{i-1} - x_iAX_i
\end{align*}
implying~$B_{i-1} = B_i + x_iAX_i$.
In particular,~$B_{m-1} = B_m + x_mAX_m$ and, inductively,~$B = B_0 = AX$.

Now suppose that~$AY=B$ for some permutation~$Y$, and let us prove that~\autoref{alg:division-permutations} does indeed return a solution.
By~\autoref{thm:detach-cycle}, we can decompose~$Y$ as~$Y = Y_1 + C_{d_1 \alcm_A B}$ for some divisor~$d_1$ of~$\minlen(A)$ coprime with~$\alcm_A B$.
Let~$B_2 = AY_1$.
By applying the same reasoning to the equation~$AY_1 = B_2$, and so on recursively, we obtain~$m$ integers~$d_1, \ldots, d_m$ such that
\begin{align*}
  Y = \sum_{i=1}^m C_{d_i \alcm_A B_i}
\end{align*}
where~$B_1 = B$ and~$B_i = A_iY_i$ with
\begin{align*}
  Y_i = Y - \sum_{j=1}^i C_{d_j \alcm_A B_j}
\end{align*}
where all~$d_i$ are divisors of~$\minlen(A)$ and coprime with~$B_i$.
Hence, by replacing each~$C_{d_i \alcm_A B_i}$ by~$d_i C_{\alcm_A B_i}$ we obtain a permutation
\begin{align*}
  X = \sum_{i=1}^m d_i C_{\alcm_A B_i}
\end{align*}
such that~$AX=B$ by~\autoref{thm:exchange-cycles}.
The permutation~$X$ maximizes the number of connected components, since its cycles are the smallest possible.
In addition~\autoref{thm:exchange-cycles} implies the uniqueness of $X$.
By grouping together all cycles of the same length, we can write
\begin{align*}
  X = \sum_{i=1}^n x_i X_i
\end{align*}
for some positive integers~$x_1, \ldots, x_n$ and \emph{distinct} cycles~$X_1, \ldots, X_n$ with~$n \le m$, and we can assume~$X_1 \lp[A] \cdots \lp[A] X_n$ without loss of generality.

It remains to show that~\autoref{alg:division-permutations} returns precisely this solution.
We prove by induction on~$i$ that the value of~$X$ after~$i$ iterations is~$\prefE{X}{i}$.
This is trivially the case for~$i=0$.
If the value of~$X$ is~$\prefE{X}{i}$ after~$i$ iterations, then~$B$ is now~$A(x_{i+1} X_{i+1} + \cdots + x_n X_n)$, hence the algorithm picks the cycle~$C_{\alcm_A A(x_{i+1} X_{i+1} + \cdots + x_n X_n)}$.
This is necessarily~$X_{i+1}$, since~$X_{i+1} \lp[A] X_j$ implies~$AX_{i+1} \lct AX_j$ and thus~$\minlen(AX_{i+1}) < \minlen(AX_j)$ whenever~$i+1 < j$.

The smallest cycle of~$A(x_{i+2} X_{i+2} + \cdots + x_n X_n)$ is longer than~$AX_{i+1}$, since~$X_{i+1} \lp[A] X_j$ implies~$AX_{i+1} \lct AX_j$ and thus~$\minlen(AX_{i+1}) < \minlen(AX_j)$ whenever~$i+1 < j$.

In order to show that the multiplicity~$x_{i+1}$ is~$\frac{b}{a}$, it suffices to show that the smallest cycle of~$A(x_{i+2} X_{i+2} + \cdots + x_n X_n)$ is longer than~$AX_{i+1}$.
Indeed, since~$AX=B$, we have~$A(x_{i+1} X_{i+1} + \cdots + x_m X_m) = B - A(x_1 X_1 + \cdots + x_i X_i)$.
Hence~$x_{i+1}AX_{i+1} \subms B - A(x_1 X_1 + \cdots + x_i X_i)$.

Since $X$ is sorted according to $\lep[A]$, it follows that $AX_{i+1} \lect AX_{i+j}$ for all $j$.
Therefore the smallest cycle length in $A(x_{i+2} X_{i+2} + \cdots + x_n X_n)$ is at least $\minLength{AX_{i+1}}$.
By contradiction, assume that  $\minLength{AX_{i+2}} = \minLength{AX_{i+1}}$, then $\aLcm{A}{A(x_{i+2} X_{i+2} + \cdots + x_n X_n)} = \aLcm{A}{A(x_{i+1} X_{i+1} + \cdots + x_n X_n)}$.
And since $\minLength{X_{i+1}} = \aLcm{A}{A(x_{i+1} X_{i+1} + \cdots + x_n X_n)}$ and $\minLength{X_{i+2}} = \aLcm{A}{A(x_{i+2} X_{i+2} + \cdots + x_n X_n)}$, we deduce that $X_{i+1} = X_{i+2}$, a contradiction.
\end{proof}

We can generalize our results to the case of equations of the form $P(X) = B$ where $B$ and all the coefficients of $A$ are permutations and $P$ is pseudo-injective.
First, we extend~\autoref{lemma:first_cond_cycle} as follows.

\begin{lemma}\label{lemma:first_cond_cycle_poly}
	Let~$B$ and~$X$ be permutations and~$P = \sum_{i=0}^{m} A_i X^i$ be a pseudo-injective  polynomial over the permutations.
	Let~$g$ be the seed of~$P$.
	Then~$P(X) = B$ implies that there exists an integer divisor~$d>0$ of~$g$ such that~$d \times \aLcm{g}{(B-A_0)}$ is an element of~$\setLength{X}$ and~$\gcd(d, \aLcm{g}{(B-A_0)}) = 1$.
\end{lemma}

\begin{proof}
	Suppose that~$P(X) = B$.
	It follows that~$P(X) - A_0 = B'$ with~$B' = B - A_0$.
	Therefore, there exists an integer~$i$ between~$1$ and~$m$ and two integers~$a$ and~$y$ respectively in~$\setLength{A_i}$ and in~$\setLength{Y^i}$ such that~$\minLength{B'} = \lcm(a,y)$.
	
	Let us begin by showing that there exists an element~$y'$ of~$\setLength{Y}$ such that~$\lcm(a,y) = \lcm(a,y')$.
	Since~$y$ is an element of~$\setLength{Y^i}$, there exist~$y_1, \ldots, y_i \in \setLength{Y}$ such that~$y = \lcm(y_1, \ldots, y_i)$.
	Therefore,~$\lcm(a,y_j) \le \lcm(a,y)$ for all~$y_j$. Since~$\setLength{Y} \subseteq \setLength{Y^i}$, we deduce that~$\lcm(a,y_j) \ge \lcm(a,y)$.
	
	Let us now prove that~$\lcm(g,y') = \lcm(a,y)$.
	Note that~$\lcm(g,y') \in \setLength{B - A_0}$, because~$y' \in \setLength{Y^j}$ for all integers~$j>0$.
	This implies that~$\lcm(g,y') \ge \lcm(a,y)$.
	However, since~$g$ is a divisor of~$a$ by definition,~$\lcm(g,y') \le \lcm(a,y') = \lcm(a,y)$.
	Therefore, the statement follows.
	From this, we deduce that~$y' = d \aLcm{g}{(B-A_0)}$.
	By a reasoning similar to the proof of~\autoref{lemma:first_cond_cycle}, we conclude that~$d$ is a divisor of~$g$ coprime with~$\aLcm{g}{(B-A_0)}$.
\end{proof}

We are now able to generalize~\autoref{lemma:second_cond_cycle} as follows.

\begin{lemma}\label{lemma:second_cond_cycle_poly}
	Let $B$ and $X$ be two permutations and $P = \sum_{i=0}^{m} A_i X^i$ be a pseudo-injective polynomial over the permutations.
	Let $g$ be the seed of $P$ and let $q>0$ be an divisor of $g$ such that $\gcd(d, \aLcm{g}{(B-A_0)}) = 1$.
	Then $P(X + \cycle{d \times \aLcm{g}{B}}) = B$ if and only if $P(X + e \cycle{(d/e) \times \aLcm{g}{B}}) = B$ for every divisor $e>0$ of $d$.
\end{lemma} 

\begin{proof}
	Let~$y = d \times \aLcm{g}{B}$ and let~$e>0$ be a divisor of~$d$.
	First, we will show that the property is true for the monomials of~$P$.
	Let~$j$ be an integer between~$1$ and~$m$ and let~$a$ be a cycle length of~$A_j$.
	Note that the definition of~$d$ and~$e$ implies that they both divide~$g$ and~$a$.
	Therefore, from the proof of~\autoref{lemma:second_cond_cycle}, it follows that~$\cycle{a} \cycle{y}^j = \cycle{y}^{j-1} \cycle{a} (e \cycle{(d/e) \times \aLcm{g}{(B-A_0)}})$.
	Applying the same substitution inductively, we deduce that~$\cycle{a} \cycle{y}^j = \cycle{a} (e \cycle{(d/e) \times \aLcm{g}{B}})^j$.
	By summing the connected components of~$A_j$, we obtain that~$A_j \cycle{y}^j = A_j (e \cycle{(d/e) \times \aLcm{g}{B}})^j$ and, by summing the different monomials of~$P$, we conclude that~$P(X + \cycle{y}) = P(X + e \cycle{(d/e) \times \aLcm{g}{B}})$.
\end{proof}

The two previous lemmas allow us to deduce that we can adapt~\autoref{alg:division-permutations} to solve equations of the form $P(X) = B$ efficiently if $P$ is pseudo-injective (\autoref{fig:run-algo} shows an execution).
This is summarized as follows.

\begin{algo}
	\label{algo:poly_cycle}
	Let~$P$ be a pseudo-injective polynomial of seed $g$ over the permutations and~$B$ a permutation.
	Let~$X = \mathbf{0}$.
	While~$B - P(X) \ne 0$ set~$X = X + C_{\alcm_g B}$.
	If at any moment~$|P(X)| > |B|$ or~$P(X)$ is not a submultiset of~$B$, then the solution is undefined.
	Otherwise, return the final value of~$X$.
\end{algo}

\begin{proposition}
	\label{prop:poly_cycle_polytime}
	The resolution of equations of the form $P(X) = B$, where $B$ is a permutation and $P$ a pseudo-injective polynomial over the permutations, is polynomial-time with respect to the sum of the sizes of the coefficients of $P$ and the size of $B$. \qed
\end{proposition}

\begin{figure*}[t]
	\centering
	\begin{tabular}{cc@{\hspace{0.5em}}c@{\hspace{0.5em}}@{\hspace{1em}}c@{\hspace{1em}}c}
		\toprule
		$B$ & $Y$ & $D$ & $P(Y+D)$ & $P(Y)$ \\
		\midrule
		$\begin{gathered}16C_2+4C_4+\\[-0.4em]18C_6+C_{12}\end{gathered}$ & $\varnothing$ & $C_1$ & $C_2+C_4+C_6$ & $\varnothing$ \\[1.5em]
		$\begin{gathered}15C_2+3C_4+\\[-0.4em]17C_6+C_{12}\end{gathered}$ & $C_1$ & $C_1$ & $4C_2+2C_4+2C_6$ & $C_2+C_4+C_6$\\[1.5em]
		$\begin{gathered}12C_2+2C_4+\\[-0.4em]16C_6+C_{12}\end{gathered}$ & $2C_1$ & $C_1$ & $9C_2+3C_4+3C_6$ & $4C_2+2C_4+2C_6$ \\[1.5em]
		$\begin{gathered}7C_2+C_4+\\[-0.4em]15C_6+C_{12}\end{gathered}$ & $3C_1$ & $C_1$ & $16C_2+4C_4+4C_6$ & $9C_2+3C_4+3C_6$ \\[1em]
		$14C_6+C_{12}$ & $4C_1$ & $C_3$ & $\begin{gathered}16C_2+4C_4+\\[-0.4em]18C_6+C_{12}\end{gathered}$ & $16C_2+4C_4+4C_6$ \\[1em]
		$\varnothing$ & $4C_1 + C_3$ &  &  & \\
		\bottomrule
	\end{tabular}
	\caption{A run of~\autoref{algo:poly_cycle} over equation~$C_2X^2 + (C_4+C_6)X = 16C_2+4C_4+18C_6+C_{12}$. 
		Each row of the table is an iteration of the loop
		This run returns the solution~$X = 4C_1 + C_3$. 
		Remark that this equation also admits~$2C_2+C_3$ and~$2C_1+C_2+C_3$ as solution, but~\autoref{algo:poly_cycle} only returns the solution which maximizes the number of connected component.}
	\label{fig:run-algo}
\end{figure*}

\section{Equations with cycles encoded compactly}
\label{sec:compact}

We will reexamine the cases of equations over permutations, but with a different encoding. 
Indeed, although until now we have represented permutations as explicit graphs (equivalent to a unary coding of the number of states in terms of size) in order to conserve the dynamical point of view, we can also represent them more compactly. 
Here we choose an encoding as arrays of pairs $(n,\ell)$ where~$\ell$ is the cycle length and $n$ is the number of copies of this cycle; if~$A$ is a permutation encoded this way, we denote by~$A[i].\ell$ the~$i$-th cycle length, and by~$A[i].n$ its number of occurrences.
We will prove that equations over permutations can be solved efficiently even under this coding, while the input size is sometimes reduced exponentially.
Let us first remark that~\autoref{algo:div_cycle} can be reused in this context.
Indeed, the computation of $\aLcm{A}{B}$ can be performed in polynomial time by~\autoref{thm:complexity-alcm}. 
In addition, the product of a permutation $A$ with a cycle $X$ with this encoding can be performed effectively.

\begin{lemma}\label{lemma:compute_prod_bin}
    Let $A$ and $X$ be two permutations encoding compactly. 
    If arithmetical operations can be executed in constant time, then computing the permutation~$AX$ can be performed in time~$\bigo{(ax)\log(ax) + \log p}$, where~$a$ and~$x$ are the lengths of the arrays encoding~$A$ and~$X$ and~$p$ the greatest cycle length in~$A$ and~$X$.
\end{lemma}

\begin{proof}
    In order to compute $AX$, we can first create an array of length $ax$ containing each cell corresponding to a product of a cell of~$A$ and a cell of~$X$.
    Indeed, $A[i] \times X[j]$~can be obtained by computing $\gcd(A[i].\ell, X[j].\ell)$ in~$\bigo{p}$ time, dividing $A[i].\ell \times X[j].\ell$ by this number in constant time in order to compute~$\lcm(A[i].\ell, X[j].\ell)$ and by multiplying~$\gcd(A[i].\ell, X[j].\ell)$ with $A[i].n \times X[j].n$.
    Therefore, this array can be computed in~$\bigo{ax + p}$ time.

    Then, we sort the array according to the cycle lengths and merge adjacent cells having the same cycle length by adding their components~$n$. 
    This can be performed in~$\bigo{(ax) \log(ax)}$ time.
\end{proof}

To more precisely analyze the complexity of~\autoref{algo:div_cycle}, we define a new process able to compute $\aLcm{a}{b}$ with reduced complexity.
This procedure is similar to Algorithm~8 in~\cite{riva2022thesis}.

\begin{algorithm}[h!]
	\begin{algorithmic}[1]
		\State $\mathit{res}_1 \gets b / a$
		\State $\mathit{res}_2  \gets 0$,
		\While{$\mathit{res}_1 \neq \mathit{res}_2 $}\label{algo_alcm:loop}
			\State $\mathit{res}_2  \gets \mathit{res}_1$
			\State $\mathit{pow} \gets \mathit{res}_1^2$
			\State $\mathit{res}_1 \gets \gcd(b,pow)$
		\EndWhile\label{algo_alcm:end_loop}
		\State \Return $\mathit{res}_1$
	\end{algorithmic}
  	\caption{$\aLcm{b}{a}$}\label{algo:aLcm}
\end{algorithm}

\begin{figure}[t]
	\centering
	\begin{tabular}{ccc}
		\toprule
		$\mathit{res}_1$ & $\mathit{res}_2 $ & $\mathit{pow}$    \\
		\midrule
		12      & 0       & 144      \\
		48      & 12      & 2304     \\
		78      & 48      & 589824   \\
		6144    & 768     & 37748736 \\
		6144    & 6144    &      \\
		\bottomrule 
	\end{tabular}
	\caption{Table of values of the variables of Algorithm~\autoref{algo:aLcm} during the calculation of $\aLcm{a}{b}$ with $b = 43008$ and $a = 3584$. 
    Each row corresponds to one iteration of the loop.}\label{table:run_alcm}
\end{figure}
An example of run of Algorithm~\autoref{algo:aLcm} is presented in~\autoref{table:run_alcm}.

\begin{lemma}\label{lemma:aLcm_poly}
	The computation of~$\aLcm{a}{b}$ is feasible with a complexity of~$\bigo{\log b (\log \log b)}$ if we assume that the four basic arithmetic operations on integers require constant time.
\end{lemma}

\begin{proof}
	Let $\prod p_i^{a_i}$ and $\prod p_i^{b_i}$ be the prime decompositions of $a$ and $b$.
    Since by definition $a$ divides $b$,  it follow that $a_i \le b_i$ for all $i$.
	We show that $\mathit{res}_1$ contains $\aLcm{a}{b}$ at the end of algorithm, in symbols
	\[
	\mathit{res}_1 = \prodC{0 \ge a_i < b_i} \qibi.
	\]
    This proof is by induction over the number iterations of the loop.
	More precisely, we prove that at the end of the $n$-th iteration of the loop $\mathit{res}_1$ is 
	\[
	\prodC{b_i \le 2^n (b_i - a_i)} \qibi \prodC{0 \le 2^n (b_i - a_i) < b_i} \qibiP{n}
	\]
    Since before the first iteration of the loop, we have  $\mathit{res}_1$ is $b/a$, the property is true.
	Let $n\ge0$ be an integer.
    Assume that the property is true for $n$.
	Then, at the start of the~$(n+1)$-th iteration $\mathit{res}_1$ is  
	\[
	\prodC{b_i \le 2^n (b_i - a_i) } \qibi \prodC{0 \le 2^n (b_i - a_i) < b_i} \qibiP{n}.
	\]
    Hence, the value of $\mathit{pow}$ is 
	\[
	\prodC{b_i \le 2^n (b_i - a_i)} \qiTbi \prodC{0 \le 2^n (b_i - a_i) < b_i} \qibiP{n+1}.
	\]
    Note that we can decompose this last value as the product
    \begin{gather*}
    	\prodC{b_i \le 2^n (b_i - a_i)} \qiTbi \\\times\\
    	\prodC{2^n (b_i - a_i) < b_i \le 2^{n+1} (b_i - a_i)} \qibiP{n+1} \\\times\\  
        \prodC{0 \le 2^{n+1} (b_i - a_i) < b_i} \qibiP{n+1}.
    \end{gather*}
    We conclude that $\gcd(b,pow)$ is
	\[
	\prodC{b_i \le 2^{n+1}(b_i - a_i)} \qibi \prodC{0 < 2^{n+1} (b_i - a_i) < b_i} \qibiP{n+1}
	\] 
    and the result follows by induction.
	
	Let $j$ be the integer such that $b_j \ge b_i$ for all $i$.
    Let  $n = \lceil\log b_j\rceil$ be a positive integer, thus $2^n{\log b_j} \ge 2^n (b_i - a_i) \ge b_j$.
    This means that after $n$ iterations of the loop the value of $\mathit{res}_1$ is $\aLcm{a}{b}$, 
    and the loop finishes after one further iteration.

    As for the complexity, note that, from the previous induction proof, the value of $\mathit{res}_1$ is always less than or equal to $b$ during each iteration.
	Therefore, the complexity of the loop body is in~$\bigo{\log b}$.
	Besides,  the number of iterations is at most $\lceil\log_2 b_j\rceil + 1$,
	and $b \ge p_j^{b_j} \ge 2^{b_j}$ implies that $b_j + 1$ is bounded by $\log_2 b + 1$.
    We conclude that the complexity of the loop is~$\bigo{\log_2 b (\log \log b)}$.
\end{proof}

Therefore, since the number of iterations of~\autoref{algo:div_cycle} is bounded by the length of the array encoding~$B$, \autoref{lemma:prod_avec_perm} and~\autoref{lemma:aLcm_poly} conclude the proof of the next result: 

\begin{lemma}\label{lemma:comp_div_cycle2}
	The complexity of~\autoref{algo:div_cycle} is
    \[\bigo{c^2 \log (cp) + c \log p (\log \log p)}\]
    where $c$ is the length of the array encoding $A$ plus that of the array encoding $B$, and $p$ is the greatest cycle length of $A$ and~$B$.
\end{lemma}

As an aside,~\autoref{lemma:comp_div_cycle2} has the consequence that, when the FDDSs of the equation~$AX=B$ are permutations represented explicitly as graphs, we can improve the runtime of~\autoref{algo:div_cycle} by first converting the input to the compact encoding described above.

\begin{corollary}
	Let $A$ and~$B$ be permutations with $A$ pseudo-cancelable.
	If $A$ and $B$ are given in input as graphs, then computing $X$ such that $AX = B$ can be performed in $\bigo{n (\log n)}$ time.
\end{corollary}

\begin{proof}
	Translating a permutation encoded as a graph into a permutation encoded by a sorted array can be accomplished by a simple traversal, therefore in $\bigo{n \log n}$ in our case. 
	Remark that the number of elements $c$ of the array is $\bigo{\sqrt{n}}$. 
	Indeed, the smallest permutation with $c$ different cycles has $\sum_{i=1}^{c} i = c(c+1)/2$ states. 
	Then, by~\autoref{lemma:comp_div_cycle2}, we conclude that the complexity of finding the solution~$X$ is $\bigo{(\sqrt{n})^2 (\log \sqrt{n} + \log n)} = \bigo{n \log n}$ time.
\end{proof}

We can now generalize~\autoref{lemma:comp_div_cycle2} to equations of the form~$P(X) = B$ with~$P$ a pseudo-injective polynomial and~$B$ a permutation.
For this, we represent the polynomial $P = A_1 X^{p_1} + \cdots + A_m X^{p_m}$, where each $A_i$ is nonzero, as an array of pairs $(p_i, A_i)$.

\begin{algo}
	\label{algo:poly_cycle2}
	Let~$P$ be a pseudo-injective polynomial of seed $g$ over the permutations and~$B$ a permutation.
	Let~$X = \mathbf{0}$.
	While~$B - P(X) \ne \mathbf{0}$ :
    \begin{itemize}
        \item perform a binary search for the number $q$ of copies of $C_{\alcm_g B}$, increasing the lower bound when~$B[0].n - r \neq 0$, with $r$ the number of copies of the smallest cycle of $P(X + q C_{\alcm_g B}) - P(X)$, and decreasing the upper bound when~$P(X + q C_{\alcm_g B})$ is not a submultiset of $B$,
        \item set~$X = X + q C_{\alcm_g B}$.
    \end{itemize}
	If at any moment~$|P(X)| > |B|$ or~$P(X)$ is not a submultiset of~$B$, then the solution is undefined.
	Otherwise, return the final value of~$X$.
\end{algo}

\begin{theorem}
	Let $P$ be a pseudo-injective polynomial over the permutations with seed $g$ and $B$ be a permutation.
    Then we solve $P(X) = B$ with complexity \[\bigo{c^2 \log(cp) \log q + c \log p (\log \log p)}\] where $c$ is the sum of the lengths of the arrays encoding the permutation of each coefficient of $P$ and $B$, $p$ is the greatest cycle length of $B$ and $q$ the greatest multiplicity of a cycle in $B$. 
\end{theorem}

\begin{proof}
	Thanks to~\autoref{prop:poly_cycle_polytime} and by applying the same reasoning as~\autoref{lemma:aLcm_poly}, we can deduce that~\autoref{algo:poly_cycle2} only outputs valid solutions. 

    Let us assume that there exists a solution of $P(X) = B$.
    We show that the algorithm does indeed return a solution.
    From our hypothesis,~\autoref{algo:poly_cycle} returns a solution $Y = \sum_{i=1}^{m_Y} y_i Y_i$ such that each $Y_i$ is a connected component and $Y_i \leprefs Y_{i+1}$ for all $i$ between $1$ and $m_Y - 1$.
    Consequently $P(Y_i) \lects P(Y_{i+1})$ for all $i$ between $1$ and $m_Y - 1$.
    We show by induction over the number of iterations of the loop that at the end of the $i$-th iteration $X$ contains $\prefE{Y}{i}$.
    Notice that this property holds before the first iteration of the loop.
    
	Let $i\ge 0$ be an integer.
    Suppose that the property is true for $i$.
    Therefore, we have $\minLength{B - P(X)} = \minLength{B - P(\prefE{Y}{i})}$.
    Hence, the value of $C_{\aLcm{g}{B}}$ computed at the $(i+1)$-th iteration is equal to $Y_{i+1}$.
    It remains to show that  $n = y_{i+1}$ with $n$ the number of copies of $C_{\aLcm{g}{B}}$ found by the algorithm.
    But this is direct since $n$ is obtained by a binary search of $y_{i+1}$.
	Indeed, by hypothesis,
	\[P(X+n \cycle{p}) - P(X) = P(\prefE{Y}{i} + n \cycle{p}) - P(\prefE{Y}{i}).\] 
	Besides, by definition of $Y$, the set of cycles having minimal length in $B - P(\prefE{Y}{i})$ comes from a product with~$Y_{i+1}$.
	Thus, if $n > y_{i+1}$ then $P(\prefE{Y}{i} + n \cycle{p}) - P(\prefE{Y}{i})$ is not a submultiset of  $B - P(\prefE{Y}{i+1})$ and if $n < y_{i+1}$ then $B - (P(\prefE{Y}{i} + n \cycle{p}) - P(\prefE{Y}{i}))$ contains some cycles of length~$\minLength{B - P(\prefE{Y}{i})}$ while it is not the case for $B - P(\prefE{Y}{i+1})$. 
    The correctness of the algorithm follows.

    To analyze the complexity, notice that the number of iterations of the loop is bounded by $c$.
	In addition, in this loop we perform a binary search whose number of iterations is bounded by  $\log q$.
    Thus, by a similar reasoning as the proof of~\autoref{lemma:comp_div_cycle2}, it follows that the operations in this binary search  have complexity $\bigo{c \log (cp)}$ and that the other operations in the loop require~$\bigo{\log d (\log \log d)}$ time.
    In total, we have thus~$\bigo{c^2 \log(cp) \log q + c \log p (\log \log p)}$ time.
\end{proof}

\begin{corollary}
    Let $P$ be a pseudo-injective polynomial over the permutations, and let $B$ be a permutation.
    If $P$ and $B$ are given as input by their explicit graphs, then the computation of a $X$ such that $P(X) = B$ can be performed with time complexity~$\bigo{n (\log n)^2}$, where $n$ is the sum of the sizes of the coefficients of $P$ and $B$. \qed
\end{corollary}

\section{Equations over the unrolls}
\label{sec:unrolls}

In this section, we explain how we can manage the transient behaviors of an FDDS in order to extend~\autoref{alg:poly-permutation}
to pseudo-injective polynomials over general FDDSs. 
For this, we will exploit the notion of \emph{unroll} in order to take into account the transient behavior of an FDDS. 

\begin{definition}[Unroll] \label{def:unroll}
	Let~$A = (A, f)$ be an connected FDDS.
	The \emph{unroll} of~$A$, denoted by~$\unroll{A}$, is the infinite graph~$G=(V,E)$ whose nodes are the pairs~$(a,k)$ of states~$a \in A$ and nonnegative integers~$k$ such that~$(a,k) \in V$ if and only if the exists a node~$u$ of the limit cycle of~$A$ such that~$f^{k}(a) = u$. 
	Moreover, the set~$E$ contains all the arcs from~$(a,k)$ to~$(a',k')$ such that~$f(a) = a'$ and~$k' = k-1$. 
\end{definition}

We usually consider unrolls up to isomorphism, in other words without node names. 
Since each connected FDDS only has one cycle, the definition of unroll implies that~$\unroll{A}$ is a forest of infinite trees where each tree has exactly one infinite branch, on which the trees representing the transient behavior of~$A$ are periodically rooted.
Remark that the unroll of a connected FDDS may contain isomorphic trees resulting from symmetries in the original graph. 
We can generalized the notion of unroll to disconnected FDDSs by summing the unrolls of its connected components;
in symbols, if~$A = A_1 + \cdots + A_n$ where each~$A_i$ is connected then~$\unroll{A} = \unroll{A_1} + \cdots + \unroll{A_n}$.

This notion was originally introduced in~\cite{article_arbre,unroll_russe} and exploited in~\cite{kroot} in order to show that any two solutions of~$AX^k = B$ 
have the same transient behavior;
that is, if 
$AX^k = B = AY^k$ then~$\unroll{X} = \unroll{Y}$. 
To prove this property, a tree product, denoted by~$\times$, such that~$\unroll{A \times B} = \unroll{A} \times \unroll{B}$ was introduced.    
Intuitively, this product is the Cartesian product applied level by level.
In order to define it, let~$\depth{v}$ be the distance of the node~$v$ from the root of the tree. 

\begin{definition}[Product of trees]\label{prodintrees}
	 Consider two trees~$\forest{t}_1=(V_1,E_1)$ and~$\forest{t}_2=(V_2,E_2)$ with roots~$r_1$ and~$r_2$, respectively. Their \emph{product} is the tree~$\forest{t}_1 \times \forest{t}_2=(V,E)$ such that 
	$V=\set{(v,u)\in V_1\times V_2 \mid \depth{v}=\depth{u}}$ and
	$E=\set{((v,u),(v',u')) \mid (v,u)\in V, (v',u')\in V, (v,v')\in E_1, (u,u')\in E_2}$.
\end{definition}
Note that the product of two trees has the minimal depth of its factors, where the depth of a tree~$\forest{t}$, denoted by~$\depth{\forest{t}}$, is the greatest length of a path between a leaf and the root. We generalize the notion of depth to forests by taking the maximum depth of its trees.
As we will often compute the product of trees with different depths, we introduce the notion of \emph{cut} of a tree~$\forest{t}$ to depth~$d$, denoted by~$\cut{\forest{t}}{d}$, as the subgraph of~$\forest{t}$ induced by the nodes of~$\forest{t}$ having depth at most~$d$.
Therefore, we have~$\forest{t}_1 \times \forest{t}_2 = \cut{\forest{t}_1}{d} \cut{\forest{t}_2}{d}$ where~$d$ the minimum of the depths of~$\forest{t}_1$ and~$\forest{t}_2$.

The set of forests with the disjoint union as addition and the above levelwise product for multiplication is a commutative semiring, with~$\mathbf{0}$ being the empty forest and~$\mathbf{1}$ the infinite path.
Notice that the set of unrolls of FDDSs with the same operations is one of its sub-semirings. 

A fundamental property which is useful when working with unrolls is that exists a total order~$\le$ on trees compatible with the product (introduce in~\cite{article_arbre}), that is,~$\forest{t}_1 \le \forest{t}_2$ if and only if~$\forest{t}_1 \forest{t} \le \forest{t}_2 \forest{t}$ for all tree~$\forest{t}$. 
By exploiting this fact, \cite{article_arbre,unroll_russe} prove the following lemma.

\begin{lemma}\label{lemma:casi_annu}
	Let~$\forest{t}_1,\forest{t}_2$ and~$\forest{t}_3$ be trees.
	Then,~$\forest{t}_1 \forest{t}_3 \le \forest{t}_2 \forest{t}_3$ if and only if~$\cut{\forest{t}_1}{\depth{\forest{t}_3}} \le \cut{\forest{t}_2}{\depth{\forest{t}_3}}$.
\end{lemma}
The order~$\le$ on trees and the previous lemma play a major role in our next proofs.

Since our goal is to solve the equation~$\sum_{i=0}^{m} \unroll{A_i} \unroll{X}^i = \unroll{B}$ efficiently with respect to the sizes of~$A_i$ and~$B$, let us point out that we cannot directly use the unrolls themselves, as they are infinite objects. 
However, we can extend~\cite[Lemmas 5.5 5.6]{kroot} in order to show that only a finite portion is actually needed.

\begin{theorem}
	Let~$A_0,\ldots,A_m$ and~$B$ be FDDSs and let~$n \ge 2 \alpha^2 + h$, where~$\alpha$ the number of nodes in the cycles of~$B$ and~$h$ the maximum depth of the trees rooted in the cycles of~$B$.
	Then~$\sum_{i=0}^{m} \unroll{A_i} \unroll{X}^i = \unroll{B}$ admits a solution if and only if~$\sum_{i=0}^{m} \cutUn{A_i} \forest{X} = \cutUn{B}$ admits a solution.
\end{theorem}  
Consequently, in order to solve the equation~$\sum_{i=0}^{m} \unroll{A_i} \unroll{X}^i = \unroll{B}$ efficiently, we just need to efficiently solve the equation~$\sum_{i=0}^{m} \cutUn{A_i} \forest{X}^i  = \cutUn{B}$ and build an FDDS from the solution of the latter equation.

Since we can use the same strategy of~\cite{kroot} to construct an FDDS~$X$ from the cut of its unroll to sufficient depth~$n$ (which can be done efficiently), it remains to prove that we can efficiently solve~$\sum_{i=0}^{m} \cutUn{A_i} \forest{X} = \cutUn{B}$.
In fact, we will prove a more general property: namely, that we can efficiently solve~$\sum_{i=0}^{m} \forest{A}_i \forest{X}^i = \forest{B}$ for arbitrary finite forests~$\forest{A}_i$ and~$\forest{B}$.
In the remainder of this section, we will prove multiple results by induction on the depth~$d$ of a forest~$\forest{F}$.
In these proofs, we do not necessarily consider~$\forest{F}$ itself, but rather a subset of~$\forest{F}$, denoted~$\talltrees{\forest{F}}{d}$, containing all the trees of~$\forest{F}$ \emph{whose depth is at least~$d$}.
Note that the definitions of sum and product imply that~$\talltrees{\forest{A} + \forest{B}}{d} = \talltrees{\forest{A}}{d} + \talltrees{\forest{B}}{d}$ and~$\talltrees{\forest{A} \times \forest{B}}{d} = \talltrees{\forest{A}}{d} \times \talltrees{\forest{B}}{d}$ for all forests~$\forest{A}$ and~$\forest{B}$ and all integers~$d$.

For our algorithm, we would like to be able to identify from~$\forest{X} = \forest{t}_1 + \ldots +\forest{t}_n$, where~$\forest{t}_i \le \forest{t}_{i+1}$, the smallest tree of~$P(\forest{X})$ originating from~$\min(\forest{A}_j)\forest{t}_i \forest{X}^{j-1}$, for some~$j$ and having the same depth as~$\forest{t}_i$, and this for all possible~$\forest{t}_i$.
The idea is that, in order to know each~$\forest{x}_i$ of~$\forest{X}$, we perform a division by~$\min(\forest{A}_j)$ to obtain a tree~$\forest{t}$, then either compute its root, if~$\forest{t} = \forest{t}_i^j$, or divide~$\forest{t}$ by the smallest tree from the already computed portion of~$\forest{X}$, raised to the power~$k-1$. 

\begin{lemma}\label{lemme:ordreProdFin1}
	Let~$d\ge 0$ be an integer and let~$P = \sum_{i=1}^{m} \forest{A}_i \forest{X}$ be a polynomial over finite forests (without~$\forest{A}_0$) such that all trees of~$\forest{A}_i$ have depth at least~$d$.
	Let~$\forest{X}$ be a forest of depth~$d$.
	
	Then, there exists an integer~$1 \le r \le m$ such that every~$\forest{x}_j$ of~$\forest{X}$ having depth~$h$ satisfies
    \begin{gather*}
        \min\big\{\forest{a}\forest{x}\forest{x}_j \mid 1 \le i \le m, \forest{a} \in \forest{A}_i, \forest{x} \in \talltrees{\forest{X}}{\depth{h}}^{i-1}\big\} =\\
        = \min(\forest{A}_r) \min(\talltrees{\forest{X}}{h})^{r-1} \forest{x}_j.
    \end{gather*}
\end{lemma}

\begin{proof}
	Let~$\forest{x}_j$ be a tree of~$\forest{X}$ of depth~$h$.
	Let~$\forest{x}_1$ be the smallest tree of~$\forest{X}$ with depth at least~$h$, in symbols~$\forest{x}_1 = \min(\talltrees{\forest{X}}{h})$.
	Let us begin by proving the case where~$\forest{x}_j = \forest{x}_1$.

    Let~$\forest{t}_{\mathrm{min}} = \min(\talltrees{P(\forest{X})}{h})$.
	Then, there exists an integer~$1 \le r \le m$ such that the forest~$\talltrees{\forest{A}_{r}}{h} \talltrees{\forest{X}}{h}^r$ contains~$\forest{t}_{\mathrm{min}}$.
	Let~$\forest{a}_1$ be the smallest tree in~$\forest{A}_r$.
	Note that, by definition,~$\forest{a}_1$ has depth greater than or equal to~$h$.
	
	Since the order is consistent with the product, it follows that~$\forest{x}_1^2 \le \forest{x}_1 \forest{x}_k$ for any tree~$\forest{x}_k$ of~$\talltrees{\forest{X}}{h}$.
	By induction, we deduce that~$\forest{t}_{\mathrm{min}}$ is a tree of~$\talltrees{\forest{A}_r}{h} \forest{x}_1^r$.
	Similarly, we also have~$\forest{a}_1 \forest{x}_1^r \le \forest{a}_{k} \forest{x}_1^r$ for all~$\forest{a}_{k} \in \talltrees{\forest{A}_r}{d}$, and we conclude that~$\forest{t}_{\mathrm{min}} = \forest{a}_1 \forest{x}_1^r$.
	The property is therefore true for~$\forest{x}_j = \forest{x}_1$.
	
	Now assume~$\forest{x}_j \neq \forest{x}_1$, implying that~$\forest{x}_j > \forest{x}_1$. 
	Using the result from the previous point, we deduce from~\autoref{lemma:casi_annu} that
	\[
	\cut{\forest{a}_1\forest{x}_1^{r-1}}{\depth{\forest{x}_1}} \le 	\cut{\forest{a}\forest{x}_1^{i-1}}{\depth{\forest{x}_1}}
	\]
	and this holds true for any integer~$1\le i\le m$ and any tree~$\forest{a} \in \forest{A}_i$.
	Note that by the definition of~$\forest{x}_1$ we have that~$\depth{\forest{x}_1} \ge h$ implying that
	\[
	\cut{\forest{a}_1\forest{x}_1^{r-1}}{h} \le 	\cut{\forest{a}\forest{x}_1^{i-1}}{h}.
	\]
	Indeed, since the order comes from a breadth-first search (\cite{article_arbre}), we are able to prove that~$\forest{a} \le \forest{b}$ if and only if~$\cut{\forest{a}}{n} \le \cut{\forest{b}}{n}$ for all integers~$n\ge 0$.
	
	Since~$\forest{x}_j$ has depth~$h$, we have~$\cut{\forest{a}_1\forest{x}_1^{r-1}}{h} \forest{x}_j = \forest{a}_1\forest{x}_1^{r-1}\forest{x}_j$ and~$\cut{\forest{a}\forest{x}_1^{i-1}}{h} \forest{x}_j = \forest{a}\forest{x}_1^{i-1}\forest{x}_j$.
	We conclude that~$\forest{a}_1\forest{x}_1^{r-1}\forest{x}_j \le \forest{a}\forest{x}_1^{i-1}\forest{x}_j$.
	And since~$\forest{x}_1^{k-1} \le \forest{x}$ for all integer~$k>0$ and with~$\forest{x} \in \talltrees{\forest{X}}{h}^{k-1}$, the lemma follows.
\end{proof}

Thanks to~\autoref{lemme:ordreProdFin1} we can show the first main result of this section.
Note that the set of finite forests of depth at most~$d_{\mathrm{max}}$ with the disjoint union for addition and the product of trees of multiplication is a semiring. 

\begin{proposition}\label{prop:injDesFinis}
	Let~$P = \sum_{i=0}^{m} \forest{A}_i \forest{X}^i$ with~$d_{\mathrm{max}} = \max_{i > 0}(\depth{\forest{A}_i)}$ and~$\depth{\forest{A}_0} \le d_{\mathrm{max}}$.
	Then,~$P$ is injective on the semiring of finite forests of depth at most~$d_{\mathrm{max}}$.
\end{proposition}

\begin{proof}
	Let~$\forest{X}, \forest{Y}$ be two finite forest of depth at most~$d_{\mathrm{max}}$ such that~$P(\forest{X}) = P(\forest{Y})$.
	Remark that~$P(\forest{X}) - \forest{A}_0 = P(\forest{Y}) - \forest{A}_0$.
	Therefore, we can assume without loss of generality that~$\forest{A}_0 = \mathbf{0}$. 
	Let~$h$ be the depth of~$P(\forest{X})$;
	then~$h \le d_{\mathrm{max}}$, as no tree of~$\forest{X}, \forest{Y}$ and of the coefficients of~$P$ have depth larger than~$d_{\mathrm{max}}$. 
	Since one of the coefficient of~$P$ contains a tree of depth~$d_{\mathrm{max}}$, it follow that~$h = \depth{\forest{X}} = \depth{\forest{Y}}$.
	Thus,~$\forest{X}$ and~$\forest{Y}$ verify the condition on depth of~\autoref{lemme:ordreProdFin1}.

	By~\autoref{lemme:ordreProdFin1}, there exist~$i, j \in \{1, \ldots, m\}$ with~$i \le j$ such that the smallest tree with factor~$\forest{x}_k$ (resp.,~$\forest{y}_k$) of~$\talltrees{\sum_{1}^{m}\forest{A_i} \forest{X}^i}{d}$ is~$\forest{a}_{i,1} \forest{x}_1^{i-1}\forest{x}_k$ (resp.,~$\forest{a}_{j,1} \forest{y}_1^{j-1}\forest{y}_k$), where~$\forest{a}_{i,1} = \min(\talltrees{\forest{A}_i}{h})$ and~$\forest{a}_{j,1} = \min(\talltrees{\forest{A}_j}{h})$. 
	We deduce that~$\forest{a}_{i,1} \forest{x}_1^i = \min(\talltrees{P(\forest{X})}{h}) =  \min(\talltrees{P(\forest{Y})}{h})= \forest{a}_{j,1} \forest{y}_1^j$ from the hypothesis.  
	
	We have~$\forest{x}_1 = \forest{y}_1$.
	By contradiction and without loss of generality, we assume that~$\forest{x}_1 < \forest{y}_1$.
	Then, it follows that~$\forest{x}_1^{j} < \forest{y}_1^{j}$.
	Since the order is compatible with the product, we deduce that~$\forest{a}_{j,1} \forest{x}_1^{j} < \forest{a}_{j,1} \forest{y}_1^j = \forest{a}_{i,1} \forest{x}_1^i$.
	Nevertheless,~$\forest{a}_{j,1} \forest{x}_1^{j} \in \talltrees{P(\forest{X})}{h}$, which is in contradiction with the minimality of~$\forest{a}_{i,1} \forest{x}_1^i$. 
	
	Finally, since~$\forest{x}_1 = \forest{y}_1$, we deduce that~$P(\forest{X}) - P(\forest{x}_1) = P(\forest{Y}) - P(\forest{y}_1)$.  
	If~$P(\forest{X}) - P(\forest{x}_1)$ contains another tree of depth~$h$, then it is also the case of~$\forest{X}$ and~$\forest{Y}$. 
	Indeed, if is not the case, all the trees of~$\forest{X} - \forest{x}_1$ and of~$\forest{Y} - \forest{y}_1$ have depth less that~$h$. 
	Therefore, since at least one of these trees appear in each product of~$P(\forest{X}) - P(\forest{x}_1)$ and of~$P(\forest{Y}) - P(\forest{y}_1)$, these two forests contain only trees of depth smaller than~$h$.  
	
	Assume that~$P(\forest{X}) - P(\forest{x}_1)$ contains a tree of depth~$h$.
	Then, by~\autoref{lemme:ordreProdFin1} and since the order is compatible with the  product, the smallest tree of~$\talltrees{P(\forest{X}) - P(\forest{x}_1)}{h}$ is~$\forest{a}_{i,1} \forest{x}_1^{i-1} \forest{x}_2$ with~$\forest{x}_2 = \min(\talltrees{\forest{X} - \forest{x}_1}{h})$.
	Likewise, the smallest tree of~$\talltrees{P(\forest{Y}) - P(\forest{y}_1)}{h}$ is ~$\forest{a}_{j,1} \forest{y}_1^{j-1} \forest{y}_2$ with~$\forest{y}_2 = \min(\talltrees{\forest{Y} - \forest{y}_1}{h})$.
	Therefore~$\forest{a}_{i,1} \forest{x}_1^{i-1} \forest{x}_2 = \forest{a}_{j,1} \forest{y}_1^{j-1} \forest{y}_2$.
	
	Since~$\forest{x}_1 =  \forest{y}_1$, it follows that~$\forest{a}_{i,1} \forest{x}_1^{i} = \forest{a}_{j,1} \forest{x}_1^{j}$  and~\autoref{lemma:casi_annu} implies that  
	\[
	\cut{\forest{a}_{i,1} \forest{x}_1^{i-1}}{h} = \cut{\forest{a}_{j,1} \forest{x}_1^{j-1}}{h}.
	\]
	Consequently and since the order is compatible with the product, we deduce that 
	\begin{gather*}
	\cut{\forest{a}_{i,1} \forest{x}_1^{i-1}}{h} \forest{x}_2 = \forest{a}_{i,1} \forest{x}_1^{i-1} \forest{x}_2 = \\
    = \forest{a}_{j,1} \forest{x}_1^{j-1} \forest{x}_2 = \cut{\forest{a}_{j,1} \forest{x}_1^{j-1}}{h} \forest{x}_2.
	\end{gather*}
	Thus,~$\forest{a}_{j,1} \forest{x}_1^{j-1} \forest{x}_2 = \forest{a}_{j,1} \forest{y}_1^{j-1} \forest{y}_2$ and since~$\forest{a}_{j,1} \forest{x}_1^{j-1}$ and~$\forest{a}_{j,1} \forest{y}_1^{j-1}$ have depth at least~$h$, we deduce from~\autoref{lemma:casi_annu} that~$\cut{\forest{x}_2}{h} = \cut{\forest{y}_2}{h}$.
	And since~$\forest{x}_2$ and~$\forest{y}_2$ have depth at most~$h$,  we conclude that~$\forest{x}_2 = \forest{y}_2$.
	
	Inductively, we have~$\talltrees{\forest{X}}{h} = \talltrees{\forest{Y}}{h}$.
	We can apply inductively the same reasoning on~$\forest{X} -  \talltrees{\forest{X}}{h}$ and~$\forest{Y} - \talltrees{\forest{Y}}{h}$, allowing us to conclude that~$\forest{X} = \forest{Y}$.
\end{proof}

\begin{corollary}\label{th:injPolyUnrolls}
	All polynomials over the unrolls are injective.
\end{corollary}

\begin{proof}
	Suppose that~$P(\unroll{X}) = P(\unroll{Y})$.
	Then, for all~$n$, we have~$\cut{P(\unroll{X})}{n} = \cut{P(\unroll{Y})}{n}$.
	If~$\unroll{X}$ and~$\unroll{Y}$ are different, there must exist an integer~$n$ such that~$\cut{\unroll{X}}{n} \ne \cut{\unroll{Y}}{n}$.
	However, by~\autoref{prop:injDesFinis} and since~$\cut{X}{n}$ are morphism, we conclude that~$\cut{\unroll{X}}{n} = \cut{\unroll{Y}}{n}$ for all~$n$.
\end{proof}

	The proof of~\autoref{prop:injDesFinis} also suggests an algorithm (inspired by~Algorithm~1 of~\cite{kroot}) for solving polynomial equations over forests of finite trees. 
	By starting from the maximal depth~$d_{\mathrm{max}}$, we can solve the equation~$\talltrees{P(\forest{X})}{d_{\mathrm{max}}} = \talltrees{\forest{B}}{d_{\mathrm{max}}}$.
	In fact, if there exists a solution~$\forest{X}$, we find it by constructing first~$\talltrees{\forest{X}} {d_{\mathrm{max}}}$. 
	In order to do this, we solve the equation~$\min(\talltrees{\forest{A}_i}{d_{\mathrm{max}}}) \forest{X}_1^i = \min(\talltrees{\forest{B}}{d_{\mathrm{max}}})$ for a certain~$i \in \{1, \ldots, m\}$.
	This can be accomplished by dividing~$\min(\talltrees{\forest{B}}{d_{\mathrm{max}}})$ by~$\min(\talltrees{\forest{A}_i}{d_{\mathrm{max}}})$ then computing the~$i$-th root of this quotient.
	The division can be made in polynomial time by~\cite{article_arbre} and the root extraction by~\cite{kroot}. 

	We obtain the smallest tree~$\forest{x}_1$ of~$\talltrees{\forest{X}}{d_{\mathrm{max}}}$, that is,
	the smallest tree of the solution having maximal depth.
	Then, by~\autoref{lemme:ordreProdFin1}, we can inductively construct~$\talltrees{\forest{X}}{d_{\mathrm{max}}}$. 
	Indeed, the~$j$-th tree~$\forest{x}_j$ of this forest is equal to
	\[
	\cfrac{\min(\talltrees{\forest{B}}{d_{\mathrm{max}}} - \talltrees{P(\forest{x}_1 + \ldots + \forest{x}_{j-1})}{d_{\mathrm{max}}} )} {\min(\talltrees{\forest{A}_i}{d_{\mathrm{max}}} \forest{x}_1^{i-1}}.
	\]
	From this, we need to find all the trees of~$\forest{X}$ whose depth is less than~$d_{\mathrm{max}}$.
	We inductively build~$\talltrees{\forest{X}}{d}$ from~$\talltrees{\forest{X}}{d'}$ where~$d < d'$ is the depth of a tree of~$\forest{B}$ and~$d' \le d_{\mathrm{max}}$ is the smallest depth of a tree of~$\forest{B}$ greater than~$d$. 
	This can be obtained by by computing~$\forest{x} = \min(\talltrees{\forest{X}}{d})$, which requires us to consider three trees:
	\begin{enumerate}
		\item~$\forest{x}_m = \min(\talltrees{\forest{X}}{d'})$, the smallest tree of~$\forest{X}$ of depth at least~$d'$,
		\item~$\forest{t} = \min(\talltrees{P(\forest{x}_m}{d})$, the smallest tree of~$P(\forest{x}_m)$ of depth least~$d$, and
		\item~$\forest{b} = \min(\talltrees{\forest{B}}{d})$, the smallest tree of~$\forest{B}$ of depth at least~$d$. 
	\end{enumerate}
	Two cases are then possible:
	either~$\forest{t}  = \forest{b}$, and then~$\forest{x}_m =  \forest{x}$,
	or there exists a coefficient~$\forest{A}_j$ such that~$\forest{a}\forest{x}^j = \forest{b}$ with~$\forest{a} = \min(\talltrees{\forest{A}_j}{d})$. 
	Therefore, it suffices to compute~$\sqrt[j]{\forest{b} / \forest{a}}$ in order to find~$\forest{x}$. 

	Now we can find the trees of~$\forest{X}$ having depth~$d$ the same way as for~$d_{\mathrm{max}}$, that is, by computing
	\[
	\cfrac{\min(\talltrees{\forest{b}}{d} - \talltrees{P(\forest{X}}{d})}{\forest{a} \forest{x}^{j-1}}.
	\]
	We can see directly that all the operations performed in this algorithm can be carried out in polynomial time.
	Furthermore, since the number of iterations is bounded by the depth of~$\forest{B}$, we deduce that this procedure has a polynomial-time complexity if we manage to find a correct coefficient in polynomial time.
	Indeed, \autoref{prop:injDesFinis}~guarantees its existence, but does not specify which coefficient should be used for each depth.
	To solve this problem, we prove the following lemma.

	\begin{lemma}\label{lemma:find_tree}
		Let~$P = \sum_{i=1}^{m} \forest{A}_i \forest{X}^i$ be a polynomial over finite forests without constant term, let~$\forest{B}$ be a finite forest and let~$d$ be the de depth of one of the trees of~$\forest{B}$. 
		If there exists~$\forest{X}$ of depth at most~$\depth{\forest{B}}$ such that~$P(\forest{X})= \forest{B}$, then the smallest tree~$\forest{x}$ of~$\talltrees{\forest{X}}{d}$ is equal to~$\sqrt[k]{\forest{b} / \forest{a}}$, where~$\forest{b} = \min(\talltrees{\forest{B}}{d'})$,~$\forest{a} = \min(\talltrees{\forest{A}_k}{d'})$,~$d'$ is the depth of~$\forest{X}$, and~$k\ge 1$ is the smallest possible integer verifying the two following conditions:
		\begin{enumerate}
			\item the result of~$\sqrt[k]{\forest{b} / \forest{a}}$ is well-defined,
			\item\label{lemma:find_tree:cond2}~$P(\forest{x})$ is a submultiset of~$\forest{b}$.
		\end{enumerate}
	\end{lemma}
	
	\begin{proof}
        Let~$\forest{t} = \sqrt[k]{\forest{b}/\forest{a}}$.
		We will show that~$\forest{x} = \forest{t}$.
		
		By the minimality of~$\forest{b}$ and by~\autoref{lemme:ordreProdFin1}, there exists a positive integer~$r$ such that~$\forest{x} = \sqrt[r]{\forest{b} / \forest{a}'}$ with~$\forest{a}' = \min(\talltrees{\forest{a}_r}{d'})$.
		We know that~$\forest{a}' \forest{x}^r = \forest{b}$ and that~$\forest{a} \forest{t}^k = \forest{b}$.
		This implies that the depth of~$\forest{t}$ is at least~$d'$.
		
		Suppose, by contradiction, that~$\forest{t} \neq \forest{x}$.
		Therefore,~$k \le r$ by hypothesis on~$k$~		However, if~$k = r$, then~$\forest{a} = \forest{a}'$, and since~$\depth{\forest{t}} = \depth{\forest{a}'} = \depth{\forest{x}}$, by~\autoref{lemma:casi_annu} and the injectivity of roots~\citep{article_arbre} we have that~$\forest{t} = \forest{x}$.
		Therefore,~$k < r$.
		This implies that~$\forest{A}_k$ and~$\forest{A}_r$ each contain at least one tree of depth at least~$d'$.
		Hence,~$\forest{B}$ has at least four trees of depth at least~$d$.
		Indeed, since~$d'$ is the depth of~$\forest{x}$, which is a tree of~$\talltrees{\forest{X}}{d}$, it follows that~$d \le d'$.

		From this, by the minimality of~$\forest{x}_r = \forest{a}' \forest{x}^r$ and since~$\forest{t}_r = \forest{a}'\forest{t}^r$ is a tree of~$\talltrees{\forest{B}}{d}$, due to the condition~\ref{lemma:find_tree:cond2}, it follows that~$\forest{t}_r \ge \forest{x}_r$.
		And since~$\forest{t} \neq \forest{x}$, we have in particular that~$\forest{t}_r > \forest{x}_r$.
		Thus, by compatibility of the order with the product, we deduce~$\forest{t}^r > \forest{x}^r$ and this implies~$\forest{t} > \forest{x}$.
		
		Now,~$\forest{t}^{k} > \forest{x}^k$ and, since the order is compatible with the product,~$\forest{b} = \forest{a} \forest{t}^k > \forest{a} \forest{x}^k$.
		This contradicts the minimality of~$\forest{b}$.
	\end{proof}   
	
	Since a coefficient verifying the two conditions of~\autoref{lemma:find_tree} can be found in polynomial time by a simple loop over the coefficients, we conclude that our algorithm is indeed efficient.
	This concludes the proof of the following theorem.
	
	\begin{theorem}\label{th:sol_poly_unroll}
		Let~$P = \sum_{i=0}^{m} A_i X^i$ be a polynomial over FDDSs and~$B$ be an FDDS.
		Then, the equation~$\sum_{i=0}^{m} \unroll{A_i} \unroll{X}^i = \unroll{B}$ can be solved in polynomial time with respect to he sum of the sizes of each~$A_i$ and~$B$.
	\end{theorem}

	An example of the execution of this algorithm is shown in~\autoref{fig:polyArbre}.

	\begin{figure*}[p]
		\centering
		\includegraphics[page=6,scale=0.75]{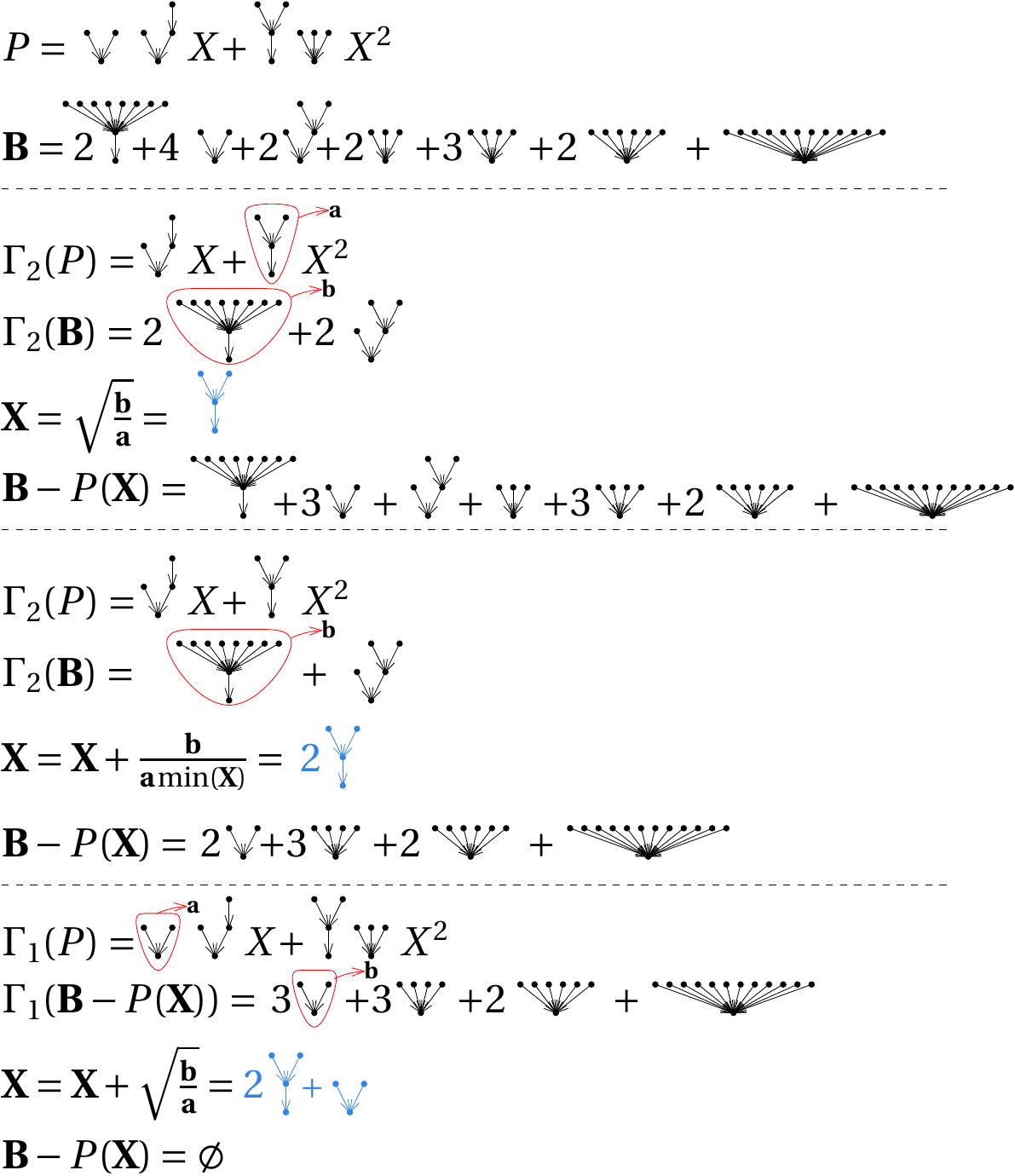}
        \bigskip
		\caption{Solving the polynomial equation~$P(\forest{X}) = \forest{B}$ over finite forests.
		The iterations of the loop are separated by dashed lines.}\label{fig:polyArbre}
	\end{figure*}

\section{Beyond permutations}
\label{sec:beyond-permutations}

Let us now consider pseudo-injective polynomials~$P = \sum_{i=0}^{m}A_i X^i$ over the FDDSs, not restricted to permutations, and equations~$P(X) = B$ where~$B$ is an arbitrary FDDS.
As in the case of permutations, the main idea is to build iteratively the solution to the equation.  
In order to do so, we will introduce a result that describes which unroll trees belong to each connected component of the solution.   
We first extend the definition of the order~$\lect$ and the order~$\lepref$ induced by it.

\begin{definition}[$\lect$ for arbitrary connected FDDSs]
	Let~$A$ and~$B$ be two connected FDDSs.
	We say that~$A \lect B$ if and only if either~$\minLength{A} < \minLength{B}$, or~$\minLength{A} = \minLength{B}$ and~$\min(\unroll{A}) \le \min(\unroll{B})$. 
	As previously, we consider that~$\varnothing \lects A$ for all nonempty connected component~$A$.
\end{definition} 

\begin{theorem}
	Let~$A$ and~$B$ be connected FDDSs.
	Then, we can check whether~$A \lect B$ in polynomial time.
\end{theorem}

\begin{proof}
	Clearly, we can compare the cycle length of $A$ with the cycle length of $B$ in polynomial time.
	Therefore, it suffices to show that we can compare $\min(\unroll{A})$ and $\min(\unroll{B})$ in polynomial times.
	For this, remark that the smallest period of  $\min(\unroll{A})$ and $\min(\unroll{B})$ are respectively at most $\minLength{A}$ and $\minLength{B}$. 
	Thus, by~Lemma~3.5 of~\cite{kroot}, we deduce that $\min(\unroll{A}) \le \min(\unroll{B}$ if and only if $\cut{\min(\unroll{A})}{n} \le \cut{\min(\unroll{B})}{n} $ with $n = \minLength{A} + \minLength{B} + \max(\depth{A},\depth{B})$.
\end{proof}

Before explaining how we can recover an unroll tree of each connected component, we introduce the following technical lemma which, given a polynomial~$P$ over the FDDSs and an FDDS~$X = X_1 + \cdots + X_{m_X}$ such that the~$X_i$ are connected and sorted according to~$\lepref$, allows us to link the smallest cycle lengths of~$P(X_{i+1})$ and~$P(X) - P(\pref{X}{i})$.

\begin{lemma}\label{lemme:give_min_length}
	Let~$P = \sum_{i=1}^{m} A_i X^i~$ be a polynomial over teh FDDSs without constant term.
	Let~$X = X_1 + \cdots + X_{m_X}$ be an FDDS sorted according to~$\lepref$. 
	Let~$i\ge0$ be an integer, and let~$D_X$ be the smallest connected component of~$P(X_{i+1})$ according to~$\lect$.
	Then,~$\minLength{D_X} = \minLength{P(X) - P(\pref{X}{i})}$. 
\end{lemma}

\begin{proof}
    Let~$A$ be a connected component of some coefficient~$A_j$ of~$P$ such that~$AX^j$ contains~$D_X$ as a connected component.
	Then~$\minLength{D_X} = \lcm(\minLength{A}, \minLength{X_{i+1}})$ thanks a reasoning similar to that of the proof of~\autoref{lemma:first_cond_cycle_poly}.
	More generally, for each component~$X_k$ between~$X_{i+1}$ and~$X_{m_X}$, there exists an integer~$a_k$ in~$\setLength{A_1 + \cdots + A_m}$ such that~$\minLength{P(X_k)} = \lcm(a_k,\minLength{X_k})$.
	The minimality of~$X_{i+1}$ implies that~$\lcm(\minLength{A}, \minLength{X_{i+1}}) \le \lcm(a_k,\minLength{X_k})$.
	Since by definition~$\lcm(a_k,\minLength{X_k}) \le \lcm(a,\minLength{X_k})$ for all~$a$ in~$\setLength{A_1+ \cdots + A_m}$, it follows that~$\lcm(\minLength{A}, \minLength{X_{i+1}}) \le \lcm(a,\minLength{X_k})$.
	Finally, the lemma follows because~$\lcm(a,\minLength{X_k}) \le \lcm(a,\minLength{X_k}, x)$ for all integers~$x>0$.
\end{proof}

We can now explain how we can recover an unroll tree of each connected component of~$X$ from~$P(X)$. 

\begin{lemma}\label{lemma:structureFDDSUnrollPoly}
	Let~$P = \sum_{i=1}^{m} A_i X^i~$ be a polynomial over the FDDSs without constant term.
	Let~$X = X_1 + \cdots + X_{m_X}$ be an FDDS sorted according to~$\lepref$, and let~$i\ge0$ be an integer.
	Then,~$\min(\unroll{X_{i+1}})$ is equal to the minimal unroll tree of~$\setDive{X-\pref{X}{i}}{\minLength{P(X) - P(\pref{X}{i})} }$. 
\end{lemma} 

\begin{proof}
	Let~$\forest{t}$ be the minimal unroll tree of~$\setDive{X - \pref{X}{i}}{\minLength{X_{i+1}}}$.
	Let~$Y$ be a connected component of~$\setDive{X - \pref{X}{i}}{\minLength{X_{i+1}}}$ that admits~$\forest{t}$ as unroll tree.
	We show that~$Y = X_{i+1}$. 
	Let~$D_X$ and~$D_Y$ be the smallest connected components, according to~$\lect$, of respectively~$P(X_{i+1})$ and~$P(Y)$.
	By definition~$X_{i+1}$ is the smallest connected component of~$\setDive{X - \pref{X}{i}}{\minLength{X_{i+1}}}$ according to~$\lepref$. 
	Note that~$\minLength{D_X} = \minLength{P(X) - P(\pref{X}{i})}$ by~\autoref{lemme:give_min_length}.
	Thus,~$X_{i+1} \lepref Y$ and therefore~$D_X \lect D_Y$. 
	We prove that~$D_X = D_Y$.

	Let~$A_j$ be a coefficient of~$P$ such that~$A_j X^j$ contains~$D_X$ as connected component. 
	Let~$A$ be a connected component of~$A_j$ such that~$AX^j$ contains~$D_X$ as connected component.
	We show that~$\minLength{D_X} = \minLength{D_Y}$.
	By hypothesis~$\minLength{A}$ is a divisor of~$\minLength{D_X}$. 
    Since~$\minLength{Y}$ is a divisor of~$\minLength{X_{i+1}}$ which is also a divisor of~$\minLength{D_X}$, we deduce that~$\lcm(\minLength{Y},\minLength{A}) \le \minLength{D_X}$.
	Moreover, the minimality of~$D_Y$ according to~$\lect$ implies~$\minLength{D_Y} \le \lcm(\minLength{Y},\minLength{A})$.
	Therefore, from the definition of~$\lect$ the assertion follows.
	Thus each connected component of~$AY^i$ has cycle length~$\minLength{D_Y}$. 
	
	We now prove~$\min(\unroll{D_X}) = \min(\unroll{D_Y})$.
	Due to~\cite{kroot}, the smallest unroll tree of~$D_X$ is~$\min(\unroll{A}) \min(\unroll{X_{i+1}})^j$. 
	Therefore since~$\forest{t} \le \min(\unroll{X_{i+1}})$ and since the order over trees is compatible with the product, we deduce~$\min(\unroll{A}) \forest{t}^j \le \min(\unroll{A}) \min(\unroll{X_{i+1}})^j$.
	Given that~$AY^j$ contains a connected component admitting~$\min(\unroll{A}) \forest{t}^j$ as an unroll tree, it result that~$\min(\unroll{D_Y}) \le \min(\unroll{D_X})$.
	Thus the definition of~$\lect$ implies that~$\min(\unroll{D_Y}) = \min(\unroll{D_X})$ and therefore~$D_X = D_Y$.
	
	Since~$X_{i+1} \lepref Y$ and since~$D_X = D_Y$, we deduce that~$X_{i+1} \lect Y$. 
	Hence, since~$\minLength{Y}$ is a divisor of~$\minLength{X_{i+1}}$, it follows that~$X_{i+1}$ and~$Y$ have the same cycle lengths. 
	Consequently~$\min(\unroll{X_{i+1}}) \le \forest{t}$. 
	Then, the minimality of~$\forest{t}$ implies~$\min(\unroll{X_{i+1}}) = \forest{t}$ and the lemma follows. 
\end{proof}  

Remark that if the polynomial~$P = \sum_{i=0}^{m} A_i X^i$ has a dendron as a connected component of one of the~$A_i$ with~$i\neq 0$, then~$X \lepref Y$ implies that~$X \lect Y$ for all connected~$X$ and~$Y$ by~\autoref{prop:inj_permutation}. 
Indeed, if we consider the smallest connected components~$D_X$ and~$D_Y$ of respectively~$P(X)$ and~$P(Y)$ according to~$\lect$, then~\autoref{prop:inj_permutation} implies~$\minLength{X} = \minLength{D_X}$ and~$\minLength{Y} = \minLength{D_Y}$. 
Therefore, if~$\minLength{D_X} < \minLength{D_Y}$ then~$X \leprefs Y$ and~$X \lects Y$.

If~$\minLength{D_X} = \minLength{D_Y}$ then there exists two integer~$i,j>0$ and two connected component~$A$ and~$A'$ from respectively~$A_i$ and~$A_j$ such that
\begin{align*}
\min(\unroll{A}) \min(\unroll{X})^i &= \min(\unroll{D_X}) \\
\min(\unroll{A'}) \min(\unroll{Y})^j &= \min(\unroll{D_Y}).
\end{align*}
This means that~$\min(\unroll{A}) \min(\unroll{X})^i  \le \min(\unroll{A'}) \min(\unroll{Y})^j$.
Now, if~$\min(\unroll{X}) > \min(\unroll{Y})$ then~$\min(\unroll{A}) \min(\unroll{X})^i > \min(\unroll{A}) \min(\unroll{Y})^i$ due to the order being compatible with the product. 
Since~$\minLength{X} = \minLength{Y}$, it follows that all the connected component of~$A Y^i$ have cycle length~$\minLength{D_X}$. 
Therefore,~$D_Y$ is not the smallest connected component of~$P(Y)$, a contradiction.
We conclude that~$\min(\unroll{X}) \le \min(\unroll{Y})$ and~$X \lect Y$. 

At this point, by combining~\autoref{prop:inj_permutation}, \autoref{lemma:structureFDDSUnrollPoly}, and~\autoref{th:cond_nes_inj} we deduce the following theorem.

\begin{theorem}
	Let~$P = \sum_{i=0}^{m} A_i X^i$ be a polynomial over the FDDSs.
	Then~$P$ is injective if and only if at least one~$A_i$ with~$i>0$ contains a dendron.
\end{theorem}

Moreover, by combining~\autoref{lemma:first_cond_cycle_poly} and~\autoref{lemma:structureFDDSUnrollPoly} we obtain the following result.

\begin{lemma}\label{lemma:first_cond_tree}
	Let~$P = \sum_{i=0}^{m} A_i X^i$ be a  pseudo-injective polynomial of seed~$g$ over the FDDSs and let~$B$ be an FDDS.
	If there exists a solution~$X = X_1 + \cdots + X_{m_X}$, where all the~$X_i$ are connected and sorted according to~$\lepref$, to the equation~$P(X) = B$, then the~$i+1$ connected components have:
	\begin{enumerate}
		\item a minimal unroll tree, with minimal period~$p$, equal to that of~$\setDive{X - \pref{X}{i}}{\minLength{B_i}}$ where~$B_i = B - P(\pref{X}{i})$,
		\item cycle length~$q \times \lcm(\aLcm{g}{B_i}, p )$ with~$q>0$ a divisor of~$\minLength{B_i} / \lcm(\aLcm{g}{B_i}, p )$ such that~$\gcd(q,\aLcm{g}{B_i}) = 1$,
	\end{enumerate}
\end{lemma}

\begin{proof}
	We prove the lemma by strong induction on the length of the prefix of~$X$.
	
	Let~$I\ge 0$ be an integer.
	Assume that the lemma is true for all integers~$0\le i \le I$. 
	Let us then show that it is true for all~$0\le i \le I + 1$. 
	For this, we only need to show that it is true for~$I+1$.
	
	By~\autoref{lemma:structureFDDSUnrollPoly}, $\min(\unroll{X_{I+1}})$ is the minimal unroll tree of~$\setDive{X - \pref{X}{i}}{\minLength{B_i}}$.
	Therefore~$X_{I+1}$ must have a cycle length multiple of~$p$, the smallest period of this tree.
	Furthermore, by~\autoref{lemme:give_min_length} we have~$\minLength{P(X_{I+1})} = \minLength{B_i}$. 
	Consequently, by a reasoning similar to~\autoref{lemma:first_cond_cycle_poly}, we deduce that~$\minLength{B_i} = \lcm(g, \minLength{X_{I+1}})$. 
	Therefore~$\minLength{X_{I+1}}$ is a multiple of~$\aLcm{g}{B_i}$.
	Thus~$\minLength{X_{I+1}} = q \times \lcm(p, \aLcm{g}{B_i})$. 
	From this, we conclude that~$q$ is a divisor of~$\minLength{B_i} /  \lcm(p, \aLcm{g}{B_i})$ and therefore, from the definition of~$\alcm$, it is coprime with~$\aLcm{g}{B_i}$, otherwise~$q \times \lcm(p, \aLcm{g}{B_i})$ would not be a divisor of~$g$.
\end{proof}

Thanks to this result, we are able to extend~\autoref{lemma:second_cond_cycle_poly} as follows:

\begin{lemma}\label{lemma:second_cond_tree}
	Let~$P = \sum_{i=0}^{m} A_i X^i$ be a pseudo-injective polynomial of seed~$g$ over the FDDSs and let~$B$ be an FDDS.
	Let~$X = X_1 + \cdots + X_{m_X}$ be an FDDS where all the~$X_i$ are connected and sorted according to~$\lepref$. 
	Then,~$P(X) = B$ if and only if~$P (X - X_i + d D) = B$ for all divisor~$d$ of~$q$, where~$q$ is a divisor of~$\minLength{B'} / \lcm(\aLcm{g}{B'}, p)$, with~$B' = B - P(\pref{X}{i-1})$, and~$X_i$ and~$D$ are two connected component such that:
	\begin{enumerate}
		\item\label{second_cond_tree:p1} $X_i$ and~$D$ have the same minimal unroll tree with minimal period~$p$;
		\item\label{second_cond_tree:p2} the cycle length of~$X_i$ is~$q \times \lcm(\aLcm{g}{B'}, p)$;
		\item\label{second_cond_tree:p3} the cycle length of~$D$ is~$(q/d) \times \lcm(\aLcm{g}{B'}, p)$.
	\end{enumerate}
\end{lemma} 

\begin{proof}
	Let~$A$ be a connected component of one of the non-constant coefficients of~$P$, and let~$a$ be its cycle length.
	Consider~$X_i$ and~$D$ as in the statement.
	To prove the lemma, it suffices to show that~$AX_i^k = A \times (d D)^k$ for all integers~$k>0$.
	Indeed, if this statement is true, then~$AX_i^k Y = A \times (d D)^k Y$ for all integers~$k>0$ and all FDDS~$Y$.
	Therefore, all the terms of the polynomial that depend on~$X_i$ will be equal to those that depend on~$d D$.
	
	Since~$X_i$ and~$D$ have the same minimal unroll tree (condition~\ref{second_cond_tree:p1}), since~$X_i$ and~$D$ are connected (conditions~\ref{second_cond_tree:p2} and~\ref{second_cond_tree:p3}), and since~$\unroll{X}$ is a morphism, it follows that~$\unroll{X_i^k} = \unroll{X_i}^k = (d \unroll{D})^k = \unroll{(d D)^k}$.
	Therefore,~$\unroll{AX_i^k} = \unroll{A \times (d D)^k}$.
	Notice that the connectedness of~$X_i$ and~$D$ implies that all connected components of~$X_i^k$ and~$(d D)^k$ have the same cycle length, namely~$\minLength{X_i}$ and~$\minLength{D}$, respectively. 
	Furthermore, since~$A$ is also connected, we deduce that all connected components of~$AX_i^k$ have a cycle of length~$\lcm(a,\minLength{X_i}) = \minLength{AX_i}$, while those of~$A (dD)^k$ have a cycle of length~$\lcm(a,\minLength{D}) = \minLength{A (dD)}$.
	To prove the assertion, it remains to show that~$\minLength{AX_i} = \minLength{A \times (d X_c)}$. 
	In other words, we need to show that~$\lcm(a,q\lcm(p,\aLcm{g}{B'})) = \lcm(a,(q/d) \times \lcm(p,\aLcm{g}{B'}))$
	
	Note that, by construction, we have~$q \times \gcd(p,a)$ and all its divisors are coprime with~$\lcm(p,\aLcm{g}{B'})/\gcd(p,a)$ and, furthermore,~$q \times \gcd(p,a)$ is a divisor of~$g$, and therefore of~$a$.
	We deduce that
	\begin{align*}
		&\phantom{=}\lcm\biggl(a,\dfrac{q}{d} \lcm(p, \aLcm{g}{B'})\biggr)\\ &= \lcm\biggl(a,\dfrac{q}{d} \times \gcd(p,a) \times \dfrac{\lcm(p,\aLcm{g}{B'})}{\gcd(p,a)}\biggr)\\
		&=\lcm\biggl(\lcm\bigg(a,\dfrac{q}{d} \times \gcd(p,a)\bigg), \lcm\bigg(a,\dfrac{\lcm(p,\aLcm{g}{B'})}{\gcd(p,a)}\bigg)\biggr)\\
		&=\lcm\biggl(\lcm\big(a,q \times \gcd(p,a)\big), \lcm\bigg(a,\dfrac{\lcm(p,\aLcm{g}{B'})}{\gcd(p,a)}\bigg)\biggr)\\
		&=\lcm\bigg(a,q \times \gcd(p,a) \times \dfrac{\lcm(p,\aLcm{g}{B'})}{\gcd(p,a)}\bigg)\\
		&=\lcm\big(a,q\times \lcm(p,\aLcm{g}{B'})\big) \qedhere
	\end{align*}
\end{proof}

By combining~\autoref{lemma:first_cond_tree} and~\autoref{lemma:second_cond_tree} we deduce the following algorithm.

\begin{algo}
	\label{alg:pseudo_inj}
	Let~$P = \sum_{i=0}^{m} A_0 X^0$ be a pseudo-injective polynomial of seed~$g$ and~$B$ be an FDDS.
	Let~$X = \mathbf{0}$ and~$P = P -A_0$ and~$B = B- A_0$ if they are well defined.
	While~$B - P(X) \ne 0$:
	\begin{itemize}
		\item find the solution~$\unroll{Y}$  of~$\sum_{i=1}^{m} \unroll{\setDive{A_i}{\minLength{B}}}\unroll{X}^i = \unroll{\setDive{B}{\minLength{B}}}$;
		\item let~$D$ the connected component of cycle length~$\lcm(p, \aLcm{g}{B})$ and minimal unroll tree~$\min(\unroll{Y} - \unroll{\setDive{X}{\minLength{B} } })$;
		\item set~$X = X + D$.
	\end{itemize}
	If at any moment~$|A| > |B|$ or~$P(X)$  is not a submultiset of~$B$ or~$\unroll{Y}$ is not defined, then the quotient is undefined.
	Otherwise, return the final value of~$X$.
\end{algo}

This algorithm is polynomial-time since the~$\alcm$ can be computed efficiently, as can all the operations involving unrolls by~\autoref{th:sol_poly_unroll}. 

\begin{theorem}\label{th:sol_poly_is_poly}
	Let~$P = \sum_{i=0}^{m} A_i X_i$ be a pseudo-injective polynomial over the FDDSs and let~$B$ be an FDDS.
	Then we can solve~$P(X) = B$ in polynomial time with respect to the sum of the sizes of the coefficients of~$P$ and the size of~$B$.
\end{theorem} 

Note that a direct consequence of~\autoref{lemma:first_cond_tree} and~\autoref{lemma:second_cond_tree} is that the equations considered in this section have a unique solution that maximizes the number of connected components, the one computed by~\autoref{alg:pseudo_inj}.
Symmetrically, there exists a unique solution that minimizes the number of connected components.

\section{Conclusions}
\label{sec:conclusions}

In this paper we have described how polynomial equations~$P(X) = B$ in one variable over the FDDSs can be solved efficiently when the polynomial is injective (also giving a characterization of this class of polynomials) and generalized this result to the larger class of pseudo-injective ones; furthermore, we have provided efficient algorithms for the compact representation of permutations as lengths of cycles in binary.

The first natural question for future work is whether there exist larger classes of polynomials in one variable over the FDDSs that admit efficient algorithms, and some classes of polynomials that do not, possibly without a constant side for the equations. Is there a suitable generalization of pseudo-injectivity, for instance with multiple seeds? Do any of our algorithms carry over to \emph{systems} of multiple polynomial equations?

Furthermore, we conjecture that the polynomials which admit the maximum number of points having the same image are among the pseudo-injective ones, as suggested by~\autoref{lemma:second_cond_tree}.

One other natural direction of research is to consider polynomials in several variables. What are, in this case, the injective ones and do they admit efficient algorithms?
Is there a natural notion of pseudo-injectivity? What are the limits of decidability?

Our result of the compact representation of permutations also suggests investigating the (presumably much higher) complexity of solving polynomial equations over FDDSs represented succinctly by Boolean circuits.

Another fundamental open problem, which is related but not directly expressible in terms of equations, is the complexity of deciding if a dynamical system is reducible, i.e., the product of two smaller nontrivial ones.

For all problems analyzed in this paper and all open problems mentioned above, it would be also interesting to consider \emph{inequalities} rather than equations, and the \emph{enumeration} of solutions rather than decision or search. Which ones admit efficient (for instance, polynomial-delay) enumeration algorithms?


\backmatter

\bmhead{Acknowledgements}

This work was partially supported by the HORIZON-MSCA-2022-SE-01 project 101131549 ``Application-driven Challenges for Automata Networks and Complex Systems (ACANCOS)'' and by the ANR project ANR-24-CE48-7504 ``ALARICE''.

\bibliography{Bibliography}

\end{document}